\documentclass[%
 reprint,
 amsmath,amssymb,
 aip,
]{revtex4-2}

\usepackage{graphicx} 
\usepackage{amsmath,amssymb,amsfonts}

\usepackage[normalem]{ulem}

\usepackage{bm}
\def \bF {\pmb{F}}

\def \bH {\pmb{H}}

\def \bx {\pmb{x}}
\def \by {\pmb{y}}
\def \bX {\pmb{X}}

\def \bz {\pmb{z}}
\def \bZ {\pmb{Z}}
\def \bw {\pmb{w}}

\def \bs {\pmb{s}}

\def \bv {\pmb{v}}
\def \bb {\pmb{b}}
\def \bp {\pmb{p}}
\def \bq {\pmb{q}}

\def \btheta {\pmb{\theta}}
\def \bbeta {\pmb{\beta}}
\def \boldeta {\pmb{\eta}}
\def \bone {\pmb{1}}
\def \diag {\text{diag}}

\usepackage{xcolor}
\usepackage{float}
\usepackage{mathtools}

\usepackage{amsthm}

\newtheoremstyle{italichead}
  {6pt}
  {6pt}
  {\normalfont}
  {}
  {\itshape}
  {.}
  {0.5em}
  {}

\theoremstyle{italichead}
\newtheorem{theorem}{Theorem}

\newtheorem{proposition}[theorem]{Proposition}
\newtheorem{corollary}[theorem]{Corollary}

\newtheorem{definition}[theorem]{Definition}

\theoremstyle{definition}
\newtheorem{remark}{Remark}

\usepackage{algorithm}
\usepackage{algpseudocode}

\usepackage{varwidth}
\usepackage{xcolor} 

\usepackage{comment}

\usepackage{enumitem}

\usepackage{xurl}

\begin{document}

\author{Amirhossein Nazerian}
\affiliation{Department of Mechanical Engineering, University of New Mexico, Albuquerque, NM 87131, USA}
\author{Sahand Tangerami}
\affiliation{Department of Mechanical Engineering, K. N. Toosi University of Technology, Tehran, Iran}
\author{Malbor Asllani}
\affiliation{Department of Mathematics, Florida State University,
1017 Academic Way, Tallahassee, FL 32306, USA}
\author{David Phillips}
\affiliation{Department of Mechanical Engineering, University of New Mexico, Albuquerque, NM 87131, USA}
\author{Hernan Makse}
\affiliation{Levich Institute and Physics Department, City College of New York, New York, NY 10031, USA}
\author{Francesco Sorrentino}
\altaffiliation{Corresponding author. Email: fsorrent@unm.edu}
\affiliation{Max Planck Institute for the Physics of Complex Systems, 01187 Dresden, Germany}
\affiliation{Department of Mechanical Engineering, University of New Mexico, Albuquerque, NM 87131, USA}



\title{Extreme vulnerability to intruder attacks destabilizes network dynamics}

\begin{abstract}
Consensus, synchronization, formation control, and power grid balance are examples of desirable dynamical states that arise in networks. Here we investigate how such states can be destabilized by an intruder agent within an otherwise functioning network. We show that a single adversarial node, coupled through adversarial connections to one or more other nodes, is sufficient to destabilize the entire network and is more effective than targeting multiple nodes. We further show that concentrating the attack on a single low-indegree node induces the greatest instability, challenging the common assumption that hubs are the most critical nodes. This leads to a new characterization of network vulnerability, identifying low-indegree nodes as the most vulnerable components. Although derived for linear systems, our results extend to nonlinear networks, including the Kuramoto model. These findings reveal an intrinsic vulnerability of technological, social, and biological networks.
\end{abstract}

\maketitle

\section{Introduction}

Complex networks are ubiquitous in technological, biological, and social systems. It is often the case that the nodes of these networks evolve over time to reach a well-defined dynamical state, such as consensus in opinion dynamics, synchronization in oscillator networks, or a desired pattern in coupled autonomous systems \cite{6341809, 9356608, 1545539}. Ensuring the robustness of such networks against failures and attacks is a critical research challenge \cite{artime2024robustness}. Robustness has been studied extensively from both structural and dynamical perspectives \cite{granovetter1978threshold, pastor2001epidemic, castellano2009statistical, centola2010spread}. Traditional approaches to studying network robustness, including structural attacks such as cascading failures via percolation theory and node or edge removals, have offered valuable insights.  Structural interventions, such as cascade failures \cite{PhysRevE.66.065102, Buldyrev_2010, PhysRevE.69.045104, motter2004cascade}, node or edge removals \cite{DENG2007714, motter2004cascade}, and network fragmentation \cite{requiao2015fast}, examine how disruptions alter the network's connectivity.

Alternatively, network instabilities may arise from the presence of adversarial connections between the nodes of a network \cite{Altafini2013Consensus,Chen2017spectral,ahmadizadeh2017eigenvalues}. In contrast to cascading failures, which are triggered by the removal of nodes or edges, adversarial interactions correspond to the addition of links with negative weights. Such interactions have been reported in brain networks as inhibitory connections between neurons \cite{ahmadizadeh2017eigenvalues}, in social networks as antagonistic relationships between agents \cite{Altafini2013Consensus}, and in distributed control systems as faulty or corrupted communication processes \cite{Chen2017spectral}.

While both cascading failures \cite{Buldyrev_2010,reis2014avoiding,albert2000error,PhysRevE.66.065102} and instabilities induced by adversarial interactions can disrupt network functionality, they operate through fundamentally different mechanisms. Cascading failures typically unfold as discrete, far-from-equilibrium processes, whereas adversarial interactions affect an existing equilibrium through continuous dynamical effects. 
In this work, we focus on the latter mechanism, namely the destabilization induced by adversarial connections. Although distinct in nature, both edge removals and adversarial interactions can ultimately lead to network instability.

A fundamental question, with significant societal and technological implications, is  how desired dynamical network states (such as coordinated dynamical states in networks of drones or autonomous vehicles) can be
disrupted by the purposeful insertion of one or a few
{adversarial} agents. 
In particular, this is relevant the 
case of cyber-attacks (CAs)  against
cyber-physical systems (CPSs) and critical infrastructures, such as power grids \cite{pagani2013power}, autonomous vehicle fleets \cite{ren2007information}, and industrial networks \cite{hearnshaw2013complex}. 
Since CAs target the cyber (software) level, while leaving the physical (hardware) level unaltered, they can be effectively modeled as network nodes that have the same structure as the others, but act adversarially.  Because CPSs rely on synchronization and coordination among such agents to ensure stability and proper functionality \cite{zhang2011optimal}, the presence of even a small number of adversarial nodes can directly undermine these mechanisms, rendering CPSs particularly vulnerable to attacks that disrupt synchronization \cite{pasqualetti2013attack}.  
Thus an essential step in order to develop resilience against  
cyber-attacks is to understand the mechanisms by which one or a few
intruder agents within 
an otherwise functioning network may compromise its dynamics. Our work also extends to attacks against social networks and biological networks, which we demonstrate using the classical Kuramoto model \cite{KuraBOOK}.

{Many real-world networks are highly directed, non-normal, and hierarchically organized, structural features known to strongly shape their dynamical response \cite{johnson2017looplessness,Asllani2018Structure,asllani2018topological,MUOLO2019Patterns,muolo2020synchronization,o2021hierarchical,Duan2022Network,nazerian2023commphys,Ramon, nazerian2024efficiency, muolo2024persistence, nazerian2026frequency}. In such systems, asymmetries in the interaction pattern generate preferred source-to-sink directions of flow, while the network architecture is often close to acyclic, with only sparse or weak feedback loops. As a consequence, perturbations do not need to spread isotropically: they can be transiently amplified, strongly biased downstream, and funneled toward terminal or sink components \cite{o2021hierarchical,Ramon,nazerian2026frequency}. These features naturally call for spectral measures that capture both stability and transient amplification when assessing vulnerability to adversarial attacks.}

{For all the cases of consensus, synchronization, formation control, power grid balance, and Kuramoto-type dynamics, we will see that the most vulnerable components of a network are the low-indegree nodes, {where the indegree of node $i$ is defined as the sum of the weights of the connections directed to node $i$,} as opposed to the hubs on which much previous work has focused \cite{artime2024robustness}. Thus 
while 
{attacks targeting the hubs usually lead to more severe cascading failures, attack strategies that leverage adversarial interactions are more effective when targeting nodes with low indegree. }

The general theme of this paper is  illustrated in Fig.\ \ref{fig:main} through an example. 
Fig.\,\ref{fig:main} a) displays a network of drones attempting to attain a desired formation.  The intruder attack consists in a single intruder node establishing adversarial connections with some of the other nodes, which is shown in b). The bottom panels of Fig.\,\ref{fig:main} are pictorial representations of the spatial trajectories of the drones before and after the attack: in c)  they converge to a set of desired target positions;  in d) they diverge from the target positions, due to the interactions with the intruder.

\section{RESULTS}

This paper investigates the case of intruder attacks on networks,
where 
our definition of an intruder is a network node that obeys the same individual dynamics  and communicates with the same output function as the other network nodes, but acts maliciously in order to disrupt the network dynamics.  In terms of the mathematical model, this means that an intruder node obeys the same individual (uncoupled) dynamics  and it uses the same output function as the other nodes in the network, but it interacts with the remaining nodes in a purposefully adversarial 
way. 
{In what follows, we mathematically formulate the problem of intruder attacks on networks and derive rigorous conditions to quantify their destabilizing effects on the dynamics. {Although our results extend to both undirected and directed networks, we focus principally on the latter as real-world systems are usually characterized by directional interactions \cite{Asllani2018Structure, o2021hierarchical, Ramon, johnson2020digraphs}.} We first focus on the important case of the linear consensus problem in networks and then consider other types of network dynamics, including formation control, power grid stability, and synchronization in oscillator networks. Our work reveals an intrinsic vulnerability in technological systems, such as autonomous networks, power grids, and sensor systems, emphasizing the need for a deeper understanding of these weaknesses and the development of specific strategies to mitigate them.}

We consider a general set of equations for the network dynamics,
\begin{equation}
    \dot{\bx}_i (t)= \bF(\bx_i (t))-\sum_{j=1}^N L_{ij} \bH(\bx_j (t) - \bbeta_j), \quad i=1,...,N, \label{general}
\end{equation}
where  $\bx_i(t)$ is the $m$-dimensional state of node $i$ at time $t$ and $N$ is the number of nodes. The individual nodal dynamics is given by $\bF(\bx_i(t))$, and the output of node $j$ is given by the function $\bH(\bx_j(t) - \bbeta_j)$ where $\bbeta_j$ is an $m$-dimensional constant vector that acts as a coordinate shift ($\bbeta_j$ may be zero). {The network connectivity is described by a digraph with adjacency matrix $A=[A_{ij}]$, where  $A_{ij}$ is the strength of the directed coupling from node $j$ to node $i \neq j$, $A_{ij}>0$ indicates a cooperative interaction (also sometimes called mutualistic or attractive), $A_{ij}<0$ indicates an adversarial interaction (also sometimes called competitive or repulsive), and $A_{ij}=0$ indicates no interaction.} The indegree of node $i$ is equal to $d_i^{in}=\sum_{j \neq i} A_{ij}$ and the outdegree of node $i$ is equal to $d_i^{out}=\sum_{j \neq i} A_{ji}$. The Laplacian matrix is denoted by $L=[L_{ij}]$ where $L_{ij}= (\delta_{ij} d_i^{in} - A_{ij})$, and $\delta_{ij}$ is the Kronecker delta. By construction, $\sum_{j=1}^N L_{ij} = 0, \forall i$. {For clarity, we note that two sign conventions for the Laplacian matrix are used in the literature: one with off-diagonal entries equal to the adjacency matrix and diagonal entries equal to the negative indegrees, and one with off-diagonal entries equal to the negative adjacency matrix and diagonal entries equal to the indegrees. In this work, we adopt the latter convention.}
We proceed under the assumption that the Laplacian matrix $L$ is `proper', i.e., all the eigenvalues of  $L$ have non-negative real parts and there is only one eigenvalue equal to $0$.
Conditions for the Laplacian matrix of a signed digraph to be proper have been investigated in the literature \cite{Altafini2013Consensus,Chen2017spectral,ahmadizadeh2017eigenvalues}.

Reference \cite[Definition 2.32]{wu2007synchronization} defines, for a given Laplacian matrix $L$, the following algebraic connectivity $f \in \mathbb{R}$ which may be considered as a generalization of the Fiedler value \cite{fiedler1973algebraic}  for digraphs:
\begin{equation} \label{eq:algeb}
    f(L) := \min_{\bX \neq \pmb{0}, \bX \perp \pmb{1}} \dfrac{\bX^\top L \bX}{\bX^\top \bX} = \lambda_{\min} \left( V^\top \dfrac{L + L^\top}{2} V \right).
\end{equation}

Here, $V \in \mathbb{R}^{N \times N-1}$ is an orthonormal basis for the null subspace of $\bone^\top$, i.e., $V$ is a matrix whose columns are normal and orthogonal to one another and have zero column sums, and $\lambda_{\min}(S)$ indicates the smallest eigenvalue of the symmetric matrix $S$.
Note that if $L$ is symmetric (the graph is undirected), the {Fiedler value of the Laplacian $\alpha = f \geq 0$,} where the strict inequality is achieved if and only if the graph is connected.
    The algebraic connectivity $f$ is relevant to studying the synchronization and consensus stabilities of directionally coupled networks. 
    Also, $f$ appears in studying the contractivity of such dynamics and their respective Lyapunov functions, e.g., see \cite[Corollary 4.23]{wu2007synchronization}. 
    Note that for $f<0$, $-{f}$ measures the time scale of the unstable dynamics.

We consider the presence of an adversarial agent, labeled as node $N+1$, and that this agent is described by the same functions $\bF$ and $\bH$ as the remaining nodes.
As a result, the network equations can be written,
\begin{equation}
    \dot{\bx}_i (t)= \bF(\bx_i (t))-\sum_{j=1}^{N+1} {L_{aug}}_{ij} \bH(\bx_j (t) - \bbeta_j),
\end{equation}
$i=1,...,N+1$, where the augmented Laplacian matrix $L_{aug}$ is either equal to
\begin{equation} \label{eq:Lnewbi}
    L_{aug}^b = \begin{bmatrix}
        L + \text{diag}(\bb) & -\bb \\
        -\bb^\top & -c
    \end{bmatrix}
\end{equation}
in the case that the intruder is bidirectionally coupled to the network nodes or equal to
\begin{equation} \label{eq:Lnewuni}
    L_{aug}^u = \begin{bmatrix}
        L + \text{diag}(\bb) & -\bb \\
        \pmb{0}^\top & 0
    \end{bmatrix}
\end{equation}
in the case that the intruder is unidirectionally coupled to the network nodes.
In both \eqref{eq:Lnewbi} and \eqref{eq:Lnewuni} the vector $\bb=[b_i \leq 0]$, $\bb \neq \textbf{0}$ (i.e., at least one $b_i<0$) contains the weights with which the additional node connects to nodes $1,...,N$. {A negative (zero) weight $b_i$ indicates the presence (absence) of an adversarial connection between the intruder and node $i$.  Furthermore, we introduce the attack budget, defined as the positive quantity $c=-\sum_{i=1}^N b_i${, which is taken to be equal to the sum of the negatives of the weights of all the connections established by the intruder with the existing network nodes.}} 

\begin{figure*}
    \centering
    \includegraphics[width=0.8\linewidth]{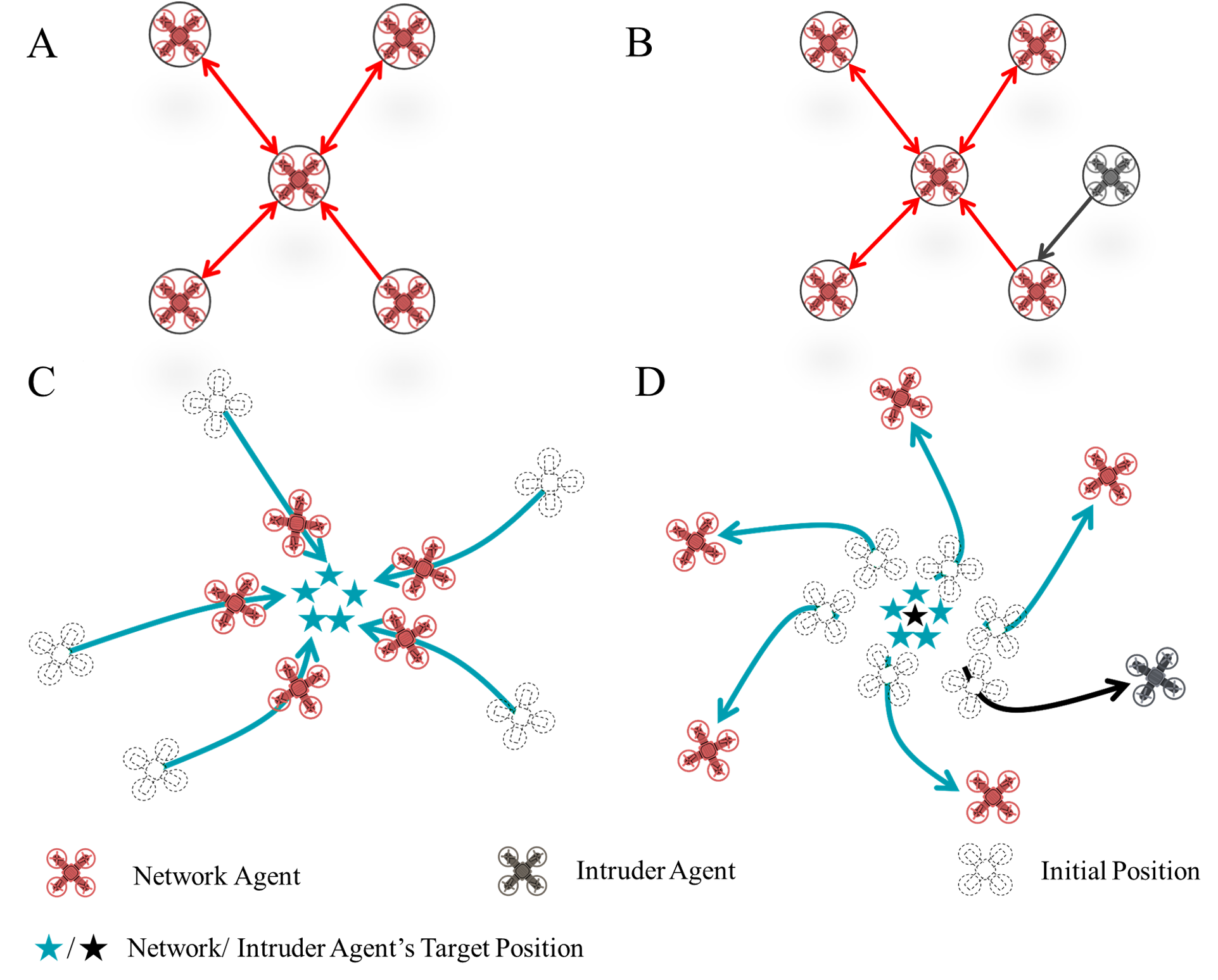}
    \caption{\textbf{Illustration of an intruder attack on a network of drones attaining a given formation.} The top panels show a network of coupled drones, before (A) and after (B) an intruder attack. The bottom panels display the trajectories of the drones before the attack in (C), showing convergence to the target positions, and after the attack in (D), showing divergence from the target position. The target positions are represented as stars in (C) and (D).}
    \label{fig:main}
\end{figure*}

{
A foundamental realization of the general dynamics \eqref{general}, which we will analyze in detail in what follows, is that of the single-integrator consensus dynamics, described by the equation,
\begin{equation} \label{consensus}
    \dot{\bX} (t) = -L \bX(t),
\end{equation}
in the $N$-dimensional state vector $\bX=[x_1,x_2,...,x_N]$ ($x_i$ here is a scalar). When the Laplacian matrix is proper, Eq.\ \eqref{consensus}
reaches consensus asymptotically independent of its initial condition, i.e., $\lim_{t \rightarrow \infty} x_i(t) = \bar{x}$, $i=1,...,N$, where $\bar{x}$ is the consensus state \cite{olfati2007consensus}. 
In Supplementary Note 1, we provide an example of single-integrator consensus dynamics \eqref{consensus} on a network subject to a single intruder attack to further motivate the relevance of the algebraic connectivity $f$. In this example, the intruder unidirectionally connects to different network nodes, resulting in different values of the algebraic connectivity $f$ as the Laplacian matrix $L_{aug}^u$ changes.
In all cases, we see that after the intruder connects to one of the other nodes, the individual trajectories $x_i (t)$ diverge from one another.

{
We characterize consensus in terms of the time evolution of the norm of the perturbations about the consensus state 
\begin{equation}
    \| \delta \bX(t) \| := \sqrt{\sum_{j=1}^{N} \left(x_j (t) - \frac{1}{N} \sum_{k=1}^{N} x_k (t)  \right)^2}. 
\end{equation} Consensus is achieved if $\|\delta \mathbf{X}(t)\| \to 0$ as $t \to \infty$. }

\begin{definition}\textit{{Destabilization of the single-integrator consensus dynamics.}}
By construction the Laplacian matrix has one eigenvalue equal zero. The consensus dynamics (6) is: (i) stable iff all the remaining eigenvalues have positive real part and (ii) unstable iff at least one of the remaining eigenvalues has negative real part. Consequently, we call a destabilization a transition from (i) to (ii).  
\end{definition}
\color{black}
We take the original Laplacian matrix $L$ to be proper, which guarantees that consensus is achieved prior to the attack. However, the addition of a single adversarial connection with the intruder implies that the resulting augmented Laplacian matrices $L_{aug}^b$ and $L_{aug}^u$ are no longer proper, as they are known to possess at least one eigenvalue with negative real part (see e.g., \cite{Chen2017spectral,ahmadizadeh2017eigenvalues}) and, as a consequence, consensus is not achieved. 


While the the addition of a single adversarial connection with the intruder produces a destabilization, in this paper we focus not on $\alpha(L_{aug})$, the asymptotic rate of instability for $\| \delta \bX(t) \|$, but rather on its transient rate of instability $|f(L_{aug})| \geq |\alpha(L_{aug})|$ immediately following the attack. As detailed in Subsection A of the Methods, a larger transient growth arises from the non-normal dynamics  characteristic of directed networks \cite{Asllani2018Structure, trefethen1991pseudospectra}. We emphasize the transient rate for three main reasons: (i) when Eq.~(1) is obtained via linearization about an unstable synchronous solution \cite{nazerian2024efficiency}, it accurately describes the dynamics only shortly after the perturbation, rendering the asymptotic rate less relevant; (ii) the transient rate provides a worst-case estimate, as it is greater than or equal to the asymptotic rate; and (iii) larger transient growth increases the likelihood that a perturbation produces a significant alteration of the system dynamics \cite{asllani2018topological, Asllani2018Structure, MUOLO2019Patterns, muolo2020synchronization, nazerian2024efficiency}.


{
Supplementary Note 1 also reviews some of the relevant properties of the algebraic connectivity $f$ from \cite{wu2007synchronization}.}}

{
In principle, one could consider the presence of more than one intruder; however, we show in what follows that under appropriate conditions, a single attacker is able to destabilize the network dynamics, so we focus on the case of a single attacker.} {The case of multiple intruders is studied in Supplementary Note 2.}

Next, we introduce our main two problems, which we call attack via bidirectional connections and attack via unidirectional connections.

{Problem 1 Statement: Attack via bidirectional connections.} Given the $N$-dimensional possibly asymmetric Laplacian matrix $L$, and a budget $-c$, find the weights $\bb = [b_i]$, $b_i \leq 0$, $i = 1, \hdots, N$, with $\sum_i b_i = -c \leq 0$ such that the algebraic connectivity $f$ of the new graph with Laplacian matrix $L_{aug}^b$ in Eq.\,\eqref{eq:Lnewbi} is minimized.

\begin{proposition} \label{prop1}
    Given $-c \leq 0$ and the Laplacian matrix $L_{aug}^b$ in Eq.\,\eqref{eq:Lnewbi}, the minimum algebraic connectivity $f$ is achieved when $b_{i^*} = -c$ for one node $i^*$ and $b_i = 0$ for all other nodes $i \neq i^*$, i.e., the entire budget must be allocated to one node.
\end{proposition}
\begin{proof}
    See Methods \ref{sec:prop1}.
\end{proof}

\begin{proposition} \label{prop2}
    Given a budget $-c \leq 0$ in Problem 1, the algebraic connectivity is $f \leq 0$.
\end{proposition}
\begin{proof}
    See Methods \ref{sec:prop2}.
\end{proof}

{Problem 2 Statement: Attack via unidirectional connections.} Given the $N$-dimensional possibly asymmetric Laplacian matrix $L$, and a budget $-c$, find the weights $\bb = [b_i]$, $b_i \leq 0$, $i = 1, \hdots, N$, with $\sum_i b_i = -c \leq 0$ such that the algebraic connectivity $f$ of the new graph with the Laplacian matrix $L_{aug}^u$ in Eq.\,\eqref{eq:Lnewuni} is minimized.

\begin{proposition} \label{prop3}
    Given $-c \leq 0$ and the Laplacian matrix $L_{aug}^u$ in Eq.\,\eqref{eq:Lnewuni}, the minimum algebraic connectivity $f$ is achieved when $b_{i^*} = -c$ for one node $i^*$ and $b_i = 0$ for all other nodes $i \neq i^*$, i.e., the entire budget must be allocated to one node.
\end{proposition}
\begin{proof}
    See Methods \ref{sec:prop3}.
\end{proof}

\begin{proposition} \label{prop4}
    Given a budget $-c \leq 0$ in Problem 2, the algebraic connectivity is $f \leq 0$.
\end{proposition}
\begin{proof}
    See Methods \ref{sec:prop4}.
\end{proof}

{Propositions \ref{prop1} and \ref{prop3} are used in the next subsections \,\ref{sec:balanced} and \ref{sec:general} to characterize the effect of the particular selection of the network nodes to which the intruder should connect, $i=1,..,N$, in order to maximize the severity of the attack.}

Supplementary Note 3 discusses the case of bidirectional attacks in undirected networks with symmetric Laplacian matrices for which we also prove that the optimal attack is to allocate the entire budget to one node.

\subsection{Balanced Directed Graphs} \label{sec:balanced}

Next, we investigate in detail the case of interest that the directed graph is balanced, i.e., for each one of its nodes, the indegree equals the outdegree, $d_i=d_i^{in}=d_i^{out}, i=1,..,N$ ($d_i$ is simply the degree of node $i$.) Undirected graphs form a particular class within the broader class of balanced directed graphs. 
{
Based on Propositions \ref{prop1} and \ref{prop3}, we know that the optimal attack is to allocate the entire budget on one connection only. 
In what follows, we study the attack to node $i$ of a given network, i.e., we set $\bb = [b_j], b_i = -c, b_{j\neq i} = 0$.
}
We are interested in how the algebraic connectivity $f$ of a balanced directed graphs is affected by an intruder attack. 
By using matrix perturbation theory, we study {Problem 1} in the limit of very small and very large budget $c$ and obtain that:
\begin{itemize}
    \item for small budget $c$: the algebraic connectivity $f(L_{aug}^u) = -0.1c$, and
    \item for large budget $c$: the algebraic connectivity $f(L_{aug}^u) = -\frac{\sqrt{2}+1}{2}c + (3+2\sqrt{2})  L_{ii}$.
\end{itemize}
For details on the derivations see Sec.\,\ref{sec:uni}.
Note the dependence on the degree of  node $i$, $d_i=L_{ii}$.


Using a similar approach, we can also study {Problem 2} and obtain that:

\begin{itemize}
    \item for small budget $c$: the algebraic connectivity $f(L_{aug}^b) = -1.1c$, and
    \item for large budget $c$: the algebraic connectivity $f(L_{aug}^b) = -2c + L_{ii}/2$.
\end{itemize}
For details of the derivations see Sec.\,\ref{sec:bi}. Also in this case, note the dependence on the degree of  node $i$, $d_i=L_{ii}$.


Figure\,\ref{fig:directed} summarizes the results of this section, by providing a schematics for how the algebraic connectivity $f$ varies with the magnitude of the budget $c$. 
The top (bottom) panel is for the case of Problem 1 (Problem 2.) 
{The figure shows that our calculations in the small and large budget limits can be used to sketch the dependence of the algebraic connectivity on $c$. In both panels, it is shown that the slopes in the small- and large-c limits are constant, while the intercept of the large-$c$ asymptote increases with the node degree. This indicates that the algebraic connectivity is minimized when the node degree is minimal.} {Let us further remark that the above result should be interpreted in the context of the adopted Laplacian. More precisely, in the appropriate cost limit, the optimal node ranking is determined by the diagonal entries $L_{ii}$ of the adopted coupling operator. For example, for the symmetrically normalized Laplacian \cite{chung1997spectral}, $L_{ii}=1$ for all $i$, and therefore this distinction between nodes vanishes.}



\begin{figure}
    \centering
    \includegraphics[width=0.8\linewidth]{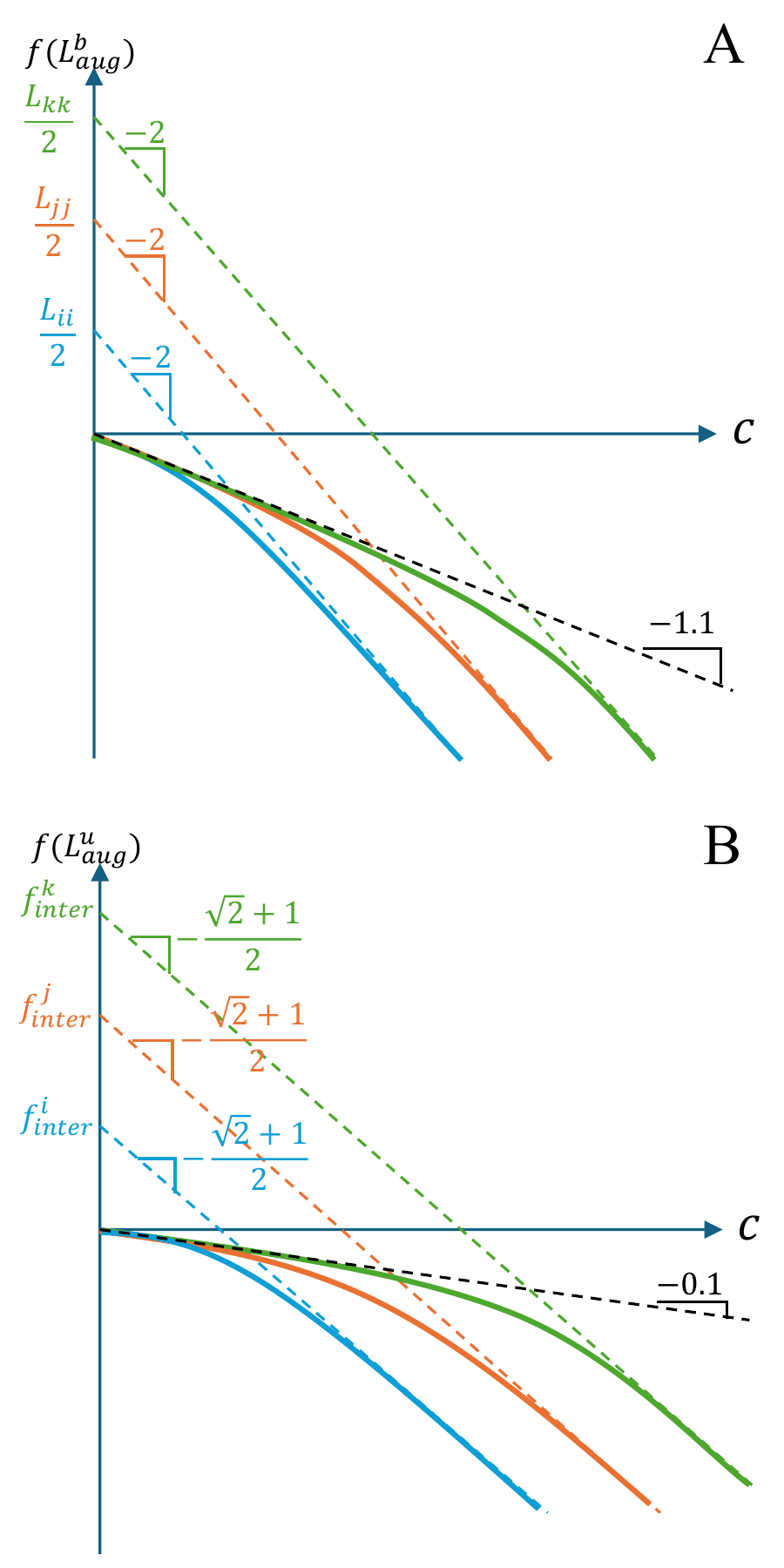}
    \caption{\textbf{Schematic dependence of the algebraic connectivity on the attack budget.}
Panels A and B schematically illustrate the algebraic connectivity, $f(L_{aug}^b)$ and $f(L_{aug}^u)$, of the asymmetric Laplacian matrices $L_{aug}^b$ and $L_{aug}^u$ defined in Eqs.~\eqref{eq:Lnewbi} and \eqref{eq:Lnewuni}, corresponding to directed graphs with bidirectional (panel A) and unidirectional (panel B) adversarial connections, respectively. In each case, a single node (either $i$, $j$, or $k$) is attacked with budget $-c<0$. The dashed black lines show the asymptotic relations $f(L_{aug}^b)=-1.1c$ (panel A) and $f(L_{aug}^u)=-0.1c$ (panel B). The colored dashed lines in panel A correspond to
$f(L_{aug}^b)=-2c+L_{xx}/2$, whereas those in panel B correspond to
$f(L_{aug}^u)=-(\sqrt{2}+1)c/2+f_{\mathrm{inter}}^x$, for $x=i,j,k$, where the intercept is
$f_{\mathrm{inter}}^x=(3+2\sqrt{2})L_{xx}$. The schematics are not drawn to scale and are intended only to illustrate the qualitative behavior.}
    \label{fig:directed}
\end{figure}

\subsection{General Directed Graphs} \label{sec:general}

The overall behavior of the algebraic connectivity $f$ for a general digraph (i.e., a digraph that is not balanced) subject to an adversarial node addition is more difficult to characterize. 
In the limit of a large budget $c$, we derive the same relationships we had previously obtained in the case of balanced digraphs: 
    $f(L_{aug}^b) = -2c + L_{ii}/2$ in the case of a bidirectional adversarial node addition and
    $f(L_{aug}^u) = -\frac{\sqrt{2}+1}{2}c + (3+2\sqrt{2})  L_{ii}$ in the case of a unidirectional adversarial node addition,
with the only difference that in this case the dependence is on the indegree $d_i^{in}=L_{ii}$ of a node, rather than its degree.
The above relations  indicate that for any directed graph, in the limit of a large budget $c$, the algebraic connectivity $f$ is defined by the indegree of the targeted node.

In general, the curve $f(L_{aug})$ vs $c$ in the limit of small $c$, i.e., the initial slope of the curve for small $c$, depends non-trivially on the choice of the targeted node. 
However, certain relations can still be derived for the average slope in directed graphs, which is discussed next.

We direct our focus to the change in the algebraic connectivity in the limit of small $c$,
\begin{equation}
    \frac{d f}{d c} := \lim_{c \to 0^-} \dfrac{\tilde{f} - f_{new}}{c}.
\end{equation}
Here, $c \geq 0$ is the budget, and $f_{new}$ ($\tilde{f}$) is the algebraic connectivity corresponding to the Laplacian matrix $L_{aug}$ before the attack, i.e., to the matrix $\tilde{L}$:
\begin{align}
\begin{split}
    \tilde{L} = &
    \begin{bmatrix}
    L & \pmb{0} \\ 
    \pmb{0}^\top &  0
    \end{bmatrix}.
\end{split}
\end{align}
We consider the two cases of (i) an attack via bidirectional connection ($L_{aug}^b$ in Eq.\,\eqref{eq:Lnewbi}), and (ii) an attack via unidirectional connection ($L_{aug}^u$ in Eq.\,\eqref{eq:Lnewuni}).
We target all nodes of the network one by one and calculate $d f / d c$ for each targeted node, i.e., we set $\bb_i = -c \pmb{e}_i$ where $\pmb{e}_i$ has all zero entries except for entry $i$ which is equal to 1.
Then, we calculate 
the average slope $< df / d c >$:
\begin{equation}
    \left< \frac{d f}{d c} \right> := \dfrac{1}{N} \sum_{i=1}^N \frac{d f_i}{d c},
\end{equation}
where $f_i$ is the algebraic connectivity when $\bb_i$ is used in either Eqs.\,\eqref{eq:Lnewbi} or \eqref{eq:Lnewuni}.

Next, we investigate the algebraic connectivity for selected real networks,
 such as animal, biological, social, neural, trade, metabolic, and genetic networks.
The network data is obtained from \cite{konect}. {All the adjacency matrices have positive entries, with certain networks being weighted and certain being unweighted (detailed information on these networks can be found in Supplementary Note 4.)} 
For each network dataset, we apply our analysis to their largest strongly connected component. 

{We emphasize that the slope $d f_i / d c$ is not easy to characterize for non-balanced digraphs since the slope $d f_i / d c$ and the indegree of the node $i$ may not have a monotonic relationship.
This is shown in Supplementary Note 4 using two examples of animal networks; for one, the slope is directly related to the indegree, while for the other, the slope is not correlated with the indegree.}

Figure\,\ref{fig:realnet} shows the mean slope $< df / d c >$, averaged over all the network nodes, for a collection of real networks under unidirectional attacks. 
We see the emergence of a clear scaling relation with the size of the network $N$, namely $< df / d c > = -1/N$. 
This result is 
analytically proven in
Sec.\,\ref{sec:meanslope}. We conclude that larger networks are less reactive in average to intruder attacks.
We have also computed the most negative slope $\min (d f / d c):= \min_i d f_i / d c$ and the mean slope $< df / d c >$ under bidirectional and unidirectional attacks, and those plots are provided in Supplementary Note 4.

\begin{figure}
    \centering
    \includegraphics[width=0.6\linewidth]{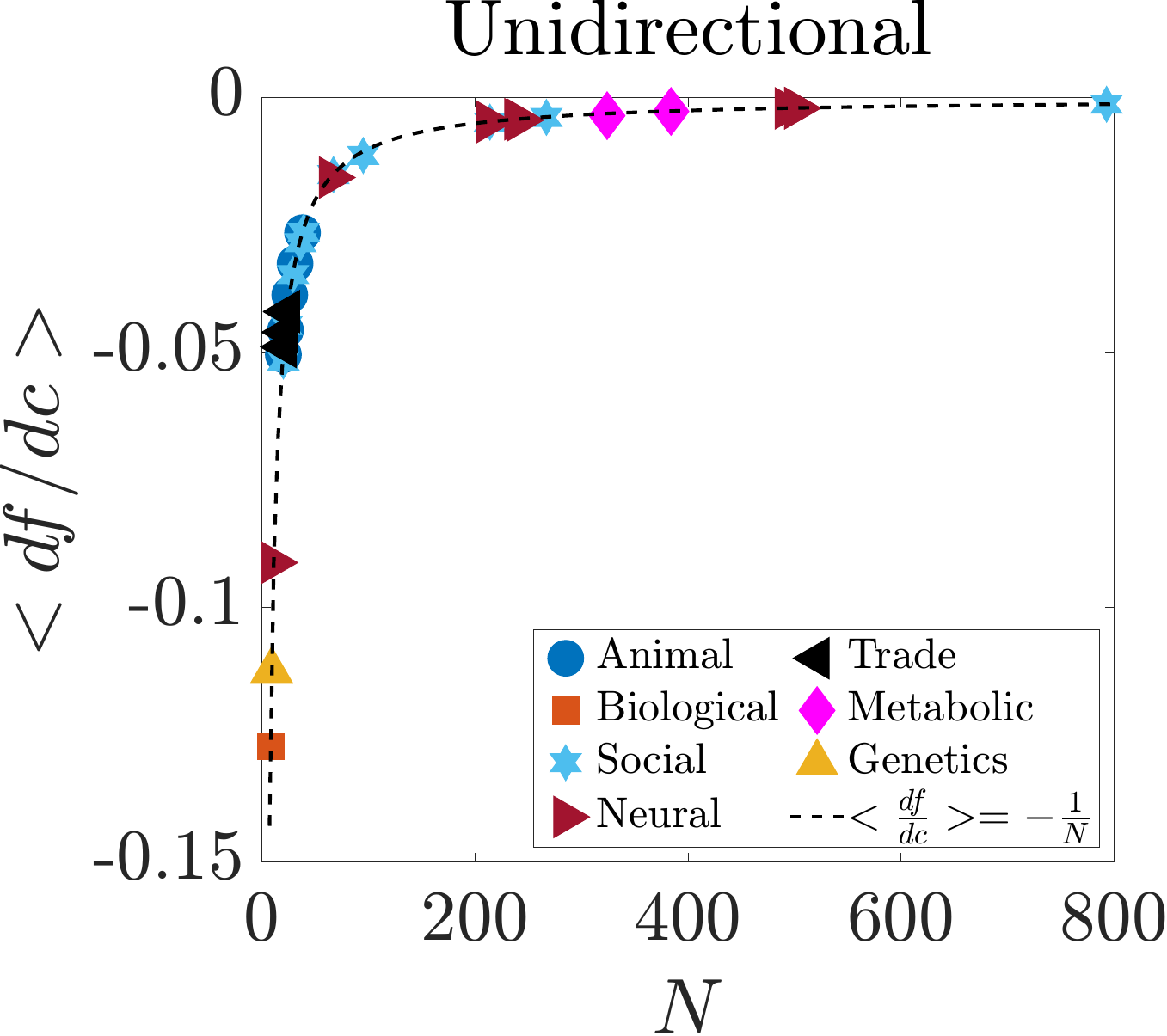}
    \caption{\textbf{Mean slope $\langle df/dc \rangle$ for unidirectional attacks on real-world networks.}
The mean slope $\langle df/dc \rangle$ for unidirectional (Eq.\,\eqref{eq:Lnewuni}) connections of the adversarial agent for selected real-world networks of size $N$. The dashed black curve shows the analytically predicted relation $\langle df/dc \rangle = -1/N$.
    } 
    
    \label{fig:realnet}
\end{figure}

\subsection{Applications} \label{sec:applications}

Although our results are derived for the case of the consensus dynamics \eqref{consensus},  they have direct implications in the case of several real-world applications, described by the general set of equations \eqref{general}.
Below we discuss three applications of interest in the realm of man-made networks. 
{A detailed asymptotic stability analysis for each application is provided in the Methods. Specifically, the synchronization dynamics is analyzed in Methods Sec.\ \ref{Sec:MAI}, the power grid dynamics in Methods Sec.\ \ref{Sec:MAII}, and the formation control dynamics in Methods Sec.\ \ref{Sec:MAIII}.}


As a first application, we first consider  synchronization of coupled chaotic oscillators.
A model for the emergence of synchronization in a network of coupled oscillators is the following,
\begin{equation} \label{eq:synch}
    \dot{\bx}_i (t) = \bF(\bx_i(t)) - \sigma \sum_{j = 1}^N L_{ij} \bH(\bx_j (t)),
\end{equation}
where $\bx_i(t) \in \mathbb{R}^m$ is the state of oscillator $i$ at time $t$, the functions $\bF: \mathbb{R}^m \to \mathbb{R}^m$ and $\bH: \mathbb{R}^m \to \mathbb{R}^m$ describe the local dynamics and the coupling functions, respectively. 
The matrix $L = [L_{ij}]$ is the Laplacian matrix that describes the network connectivity, and the scalar $\sigma \geq 0$ is the coupling strength. For this example, we set  the local dynamics $\bF$ to be the chaotic Lorenz oscillator. The functions $\bF$ and $\bH$ are
\begin{equation} \label{eq:Lorenz}
        \bx = \begin{bmatrix}
        x \\ y \\ z
    \end{bmatrix}, \ \
    \bF(\bx) = \begin{bmatrix} 10(y-x) \\
    x(28 - z) - y \\
    xy - 2 z \end{bmatrix}, \ \  \bH(\bx) = \begin{bmatrix}
        0  \\
        y  \\
        0 
    \end{bmatrix}.
\end{equation}

We take the directed network shown in Fig.\,\ref{fig:application} A and set $\sigma=6$.
Figure \ref{fig:application}B shows the synchronous time evolution of the state components $x_i(t)$ for each oscillator before the attack. Panel C shows the loss of synchronization observed after the attack.
The loss of stability is the direct consequence of having mixed sign eigenvalues of the Laplacian matrix after the attack. 
The plots of the $y$ and $z$ components show similar behavior as the $x$ component and can be found in Supplementary Note 5. 
In Supplementary Note 5 we also provide an explanation for why the synchronous state becomes unstable when an attack occurs in terms of the {master stability function approach (MSF) \cite{fujisaka1983stability,pecora1998master}}.

Reference \cite{sorrentino2007controllability} has used the master stability function to study stability of the synchronous solution in the case of the pinning control problem, for which a reference synchronous trajectory is selected to which all the network nodes need to synchronize. 
{There are, however, two main differences between the two approaches.} { {The first concerns the setting of the problem:} in pinning control one {considers} the role of a reference agent in driving the synchronization of a network of coupled systems {toward} the reference trajectory, whereas in our work {we consider} the role of an intruder agent in disrupting synchronization. In the former case the coupling of the reference agent to the rest of the network is cooperative, 
while in the latter it is adversarial. 
{In particular, the insertion of an intruder effectively alters the coupling structure so that the relevant Laplacian has at least one negative part eigenvalue; as a result, the asymptotic stability analysis naturally involves the evaluation of the Master Stability Function at negative real arguments, corresponding to these eigenvalues, which is not the common approach in the literature. } The Methods Sec.\ I provides master stability function plots, over both negative and positive values of its argument, for the cases of the Lorenz, Rossler, and Chua oscillators, under various node-to-node couplings. For all the cases considered, we see that the MSF is positive when its argument is negative, implying instability of the synchronous solution.
{The second difference concerns the observable used to characterize the effect of the intervention: while Ref.\ \cite{sorrentino2007controllability} focuses on asymptotic stability through Lyapunov exponents, here we focus on the transient rate of instability. }} 



\begin{figure*}
    \centering
    \includegraphics[width=0.8\linewidth]{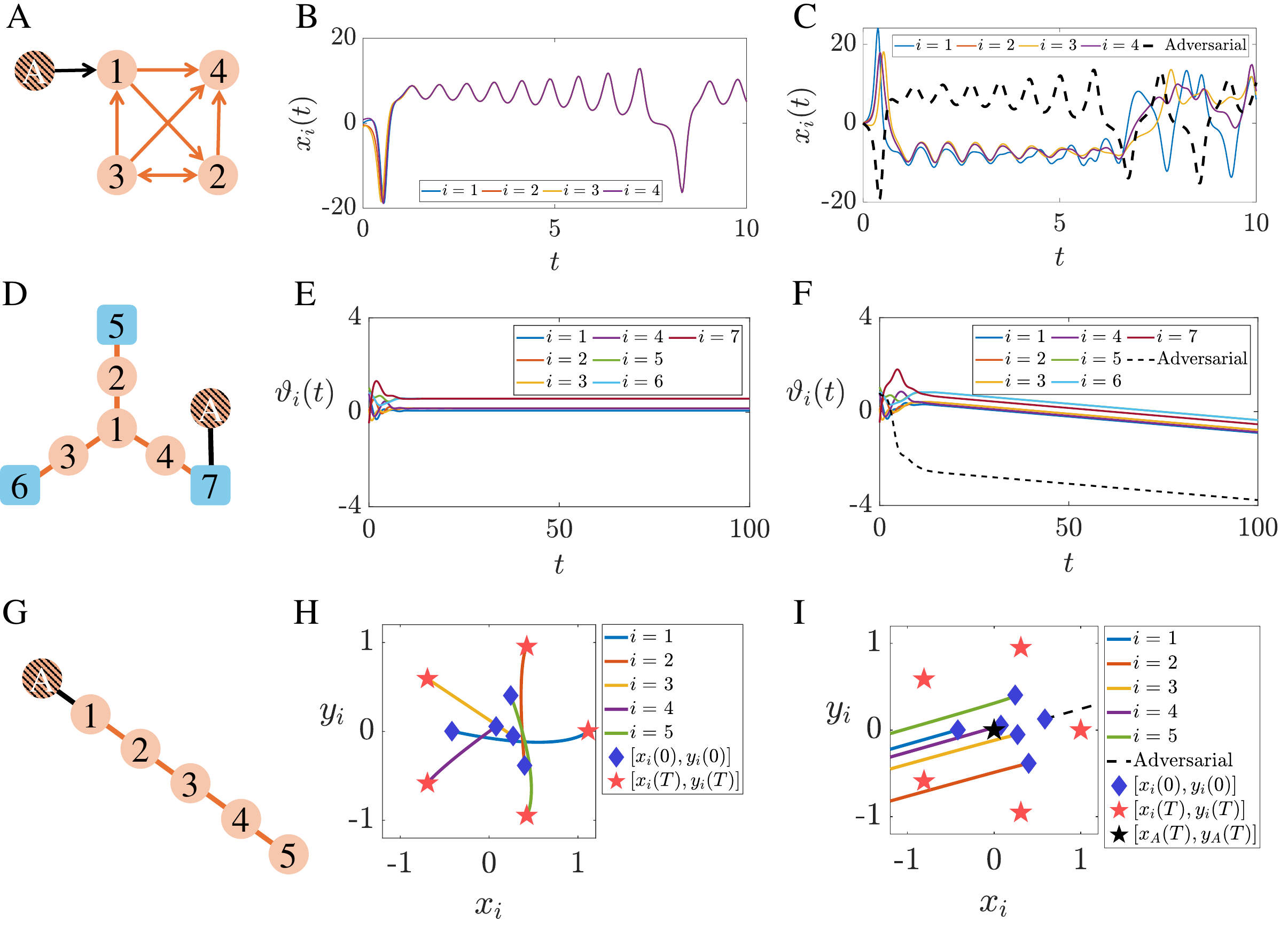}
    \caption{
    \textbf{Effects of intruder attacks in different applications.}
    The top row (A--C) illustrates the application to directly coupled chaotic oscillators, the middle row (D--F) to the power grid described by the nonlinear swing equation, and the bottom row (G--I) to linear single-integrator robot formation control. The left column (A, D, G) shows the corresponding network topologies. Panel A displays a directed network, whereas panels D and G show undirected networks. In all networks, black and orange links denote edge weights of $-1$ and $1$, respectively. The middle column (B, E, H) shows the dynamics in the absence of the adversarial agent (node A), while the right column (C, F, I) shows the dynamics in its presence.
In panel D, blue square nodes represent generators and orange circular nodes represent load nodes; the adversarial agent is also a load node. Panels B and C show the states $x_i(t)$ of the Lorenz oscillators; panels E and F show the angular positions $\vartheta_i(t)$ of the generators and loads in the swing equation; and panels H and I show the robot positions in the $xy$ plane, where diamonds denote the initial positions and stars denote the desired final positions.
    }
    \label{fig:application}
\end{figure*}

Next, we study the effect of a unidirectional attack on a scale-free network of Lorenz oscillators coupled in the $y$-variable, Eq.\ \eqref{eq:Lorenz}. {For this case, the MSF analysis for Eq.~\eqref{eq:Lorenz}, presented in Sec.\ I of the Methods, shows that, for negative values of its argument, the master stability function becomes increasingly negative as the magnitude of the argument increases.}
We use the algorithm from \cite{goh2001universal} to randomly generate an undirected and unweighted scale-free graph, which is shown in Fig. \ref{fig:synch_node} A. 
The network has $N=100$ nodes and the degree distribution follows a power law $P_D(k) \propto k^{-\gamma}$ with $\gamma = 2.33$ with an average degree of $10$.
We aim to study the effect of the attack on different nodes; we set $\sigma = 10$, $c = 1$, and vary the node $i$ that is targeted.
We select the nodes $i = 1,5,17,94,100$ 
with degrees $82, 47, 23, 13$, and $8$, respectively, and associate to these nodes the colors blue, orange, yellow, purple, and green, respectively. 
For more details on this example, see Supplementary Note 6.

The perturbation vector for each node is calculated as $\delta \bx_i(t) = \bx_i (t) - \frac{1}{N} \sum_{j = 1}^N \bx_j \in \mathbb{R}^3$. 

We then form the transverse perturbations vector $\delta \bX(t) = [\delta \bx_1(t)^\top, \ \delta \bx_2(t)^\top, \hdots, \delta \bx_{N}(t)^\top]^\top \in \mathbb{R}^{3N}$ and calculate its norm:
\begin{equation} \label{eq:trans_synch}
    \| \delta \bX(t) \| := \sqrt{\delta \bX(t) ^\top \delta \bX(t)}.
\end{equation}
\color{black}
If the oscillators synchronize, then $\| \delta \bX(t) \| \to 0$, otherwise $\| \delta \bX(t) \| \gg 0$.

We see from Fig.~\ref{fig:synch_node} B that before the attack, i.e., on the left side of the dashed black line, the norm of the transverse perturbations approaches zero, indicating that the Lorenz oscillators have synchronized. However, after attacking either one of the nodes $1,5,17,94,100$,
we see an immediate loss of synchronization, with a faster rate of transient instability when the nodes with lower degrees are targeted. For example, the attack on node 100 with degree 8 (the green curve) results in much larger $\| \delta \bX(t) \|$ than the attack on node 1 with degree 82 (the blue curve).
To conclude, this example shows that in a nonlinear networked system, lower indegree nodes are more vulnerable to intruder attacks. {Note that this is consistent with the MSF analysis for the case of Eq.~\eqref{eq:Lorenz}, presented in Sec.\ I of the Methods, since attacking nodes with lower indegrees corresponds to evaluating the master stability function at more negative values of its argument, for which the associated maximum Lyapunov exponent is larger.}



\begin{figure}
    \centering
    \includegraphics[width = \linewidth]{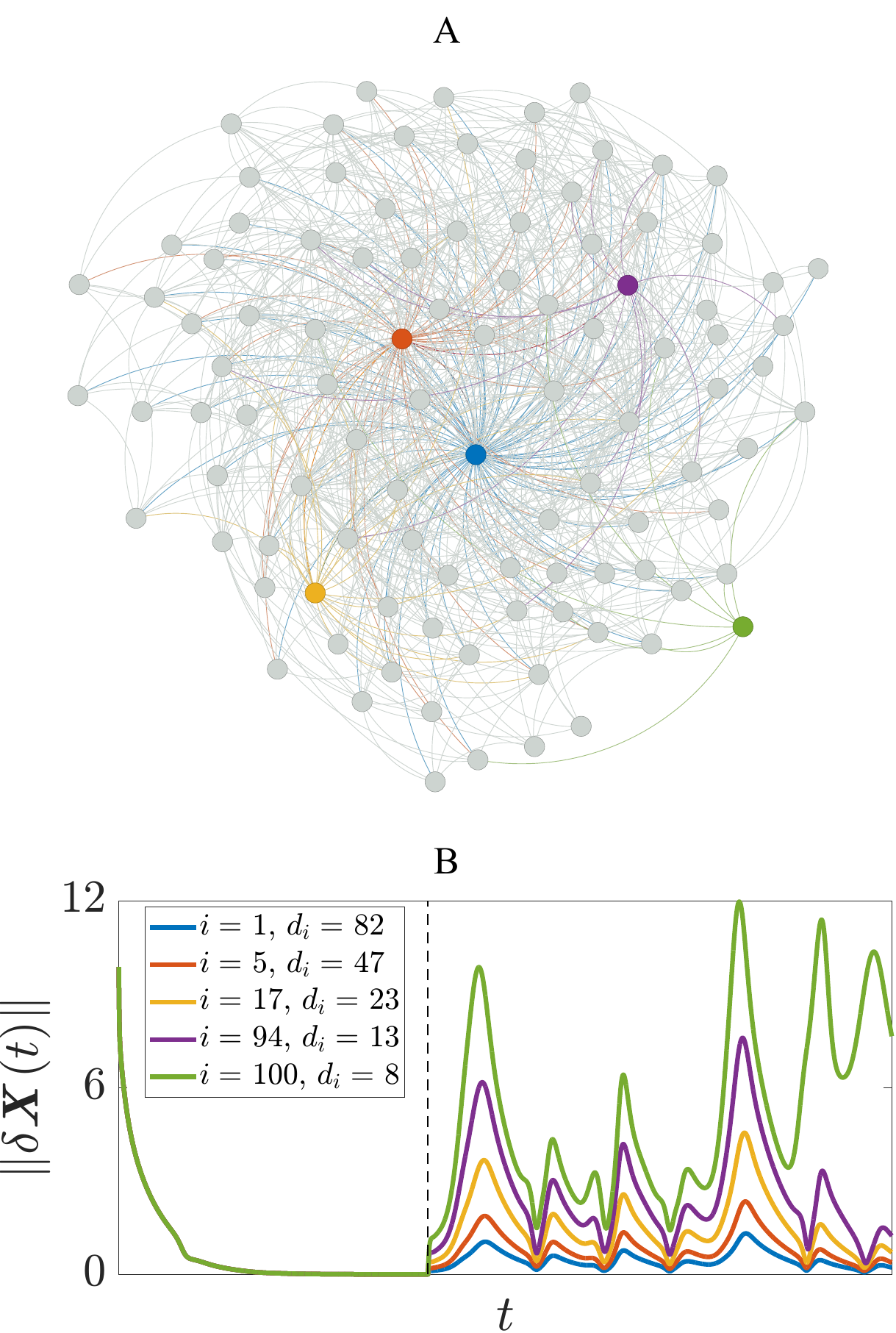}
    \caption{\textbf{Effect of the attacked node on network destabilization.}
Panel A shows a randomly generated scale-free network under attack through either of the colored nodes (blue $i=1$, orange $i=5$, yellow $i=17$, purple $i=94$, and green $i=100$). Panel B considers the dynamics of a network of coupled Lorenz oscillators, Eq.~\eqref{eq:synch}, and plots the norm of the transverse perturbations $\|\delta \mathbf{X}(t)\|$ as a function of time $t$ for attacks targeting different nodes $i$. Each curve corresponds to a different attacked node, all with the same attack budget $c=1$. The dashed black line indicates the onset of the attack at $t=4\,\mathrm{s}$. The initial conditions of the oscillators are identical in all simulations, so that the observed differences are solely due to the attacked node. The degree of each attacked node, $d_i$, is also indicated in the plot.}
    \label{fig:synch_node}
\end{figure}

As a second application we consider power grid dynamics. The swing equation is  a classical model for the dynamics of a power grid, \cite{nishikawa2015comparative,bhatta2021modal}
\begin{equation} \label{eq:swingsimple}
     \ddot{\vartheta}_i (t) = p_i - {\gamma} \dot{\vartheta}_i (t) - \sum_{j= 1, j \neq i}^N A_{ij} \sin(\vartheta_i (t) - \vartheta_j (t))
\end{equation}
where $\vartheta_i (t)$ is the angular displacement of node/rotor $i$, the scalar $p_i$ denotes the power consumption ($p_i <0$ for loads) or power generation ($p_i >0$ for generators,) the scalar ${\gamma >0}$ denotes the damping ratio, and the symmetric adjacency matrix $A = [A_{ij}]$ describes the connectivity of the grid, i.e.,  $A_{ij}=A_{ji}=1$ ($0$) if there is (is not) a transmission line between between nodes $i$ and $j$. 
{Note that although Eq.\ \eqref{eq:swingsimple} is not a particular case of Eq.\ \eqref{general}, the linearized swing equation (see Supplementary Note 7) on which our analysis is based is.}



{We emphasize that the swing equation \eqref{eq:swingsimple} conveniently demonstrates how an intruder oscillator can be mathematically incorporated into the Laplace matrix through a negative sign. By constraining the intruder oscillator to remain in anti-phase with the rest of the network—i.e., \(\sin(\vartheta_i(t) - \vartheta_j(t) \pm \pi) = -\sin(\vartheta_i(t) - \vartheta_j(t))\)—a negative sign naturally emerges in the coupling kernel. This negative sign is then absorbed into the adjacency matrix as a negative link connection, represented by \(-A_{ij}\).} 

To illustrate the effect of the addition of an adversarial agent, we provide a numerical example. 
Here, we set 
$\gamma = 0.9$ and the network shown in Fig.\,\ref{fig:application} D shows the grid topology before and after the addition of the adversarial agent. 
Nodes $1,2,3$ and $4$ are loads, nodes $5,6$ and $7$ are generators, and the adversarial node is denoted as node A with a black connection with weight $-1$.
The vector $\bp = [
-0.3, \ -0.3, \ -0.3, \ -0.3, \ 0.4, \ 0.4, \ 0.4 
]^\top$ describes the power generated/consumed in the original power grid.
Without loss of generality, for the adversarial agent, we set $p_{A} = -0.1$, so the adversarial agent acts as a load.

As seen in Fig.\,\ref{fig:application} E (Fig.\,\ref{fig:application} F), the dynamics of the power grid converges (does not converge) to a fixed point before (after) the attack.
In Supplementary Note 7, we show analytically how the addition of the adversarial agent results in at least one unstable mode of the linearized swing equation, which in turn destabilizes the nonlinear dynamics.

As a third application we consider  distributed formation control. 
The formation control of robots is used in both civilian and military applications, see \cite{oh2015survey,wang2021formation,yu2022decentralized}.
Here, we focus on positioning a set of robots  on a given desired pattern, e.g., the robots are positioned around a circle with a known diameter.

We take the shape-based formation control strategy from \cite{Morbidi2022Functions}:
\begin{equation} \label{Eq_formation_control}
    \dot{\bX} (t) = (- F(L) \otimes I_2 ) (\bX(t) - \bbeta),
\end{equation}
where the state vector $\bX(t) = [x_1(t), y_1 (t), x_2 (t), y_2 (t), \hdots, x_N (t), y_N (t)]^\top \in \mathbb{R}^{2N}$ is the concatenation of the $xy$ coordinates of $N$ robots at time $t$.
The matrix $L$ is the symmetric Laplacian matrix that describes the connectivity between the robots, the function $F(L): \mathbb{R}^{N \times N} \to \mathbb{R}^{N \times N}$ is a function of the Laplacian matrix $L$.
The network topology for this example with $N=5$ is shown in Fig.\,\ref{fig:application} G.
We set $F(L) = I_N - \exp (-3L)$ as in  \cite{Morbidi2022Functions}.
The vector $\bbeta \in \mathbb{R}^{2n}$ is a constant vector of target positions, which we
 position around a circle with radius $1$ in equidistant phases, i.e., $\bbeta = [\sin (0), \cos (0), \sin (\frac{2 \pi}{5}), \cos (\frac{2 \pi}{5}),\hdots, \sin (\frac{8 \pi}{5}), \cos (\frac{8 \pi}{5})]^\top$.
For the adversarial agent, we set its corresponding $\bbeta$ to the origin $[0,0]^\top$.

Figure \ref{fig:application} H shows the formation control in the absence of the adversarial agent, where each robot converges to its respective target position in less than $20$ seconds. Note that this can be seen as a case of cluster synchronization, since different agents converge to different states.
However, Fig.\,\ref{fig:application} I shows that in the presence of the adversarial agent, the robots no longer reach the target positions and the closed-loop system becomes unstable, position and velocity of the robots tend to infinity.)

\subsection{Phase synchronization based on Kuramoto dynamics}

Here we consider the case of phase synchronization in a network of coupled Kuramoto oscillators \cite{KuraBOOK},
\begin{equation} \label{eq:Kuramoto}
    \dot{\theta}_i (t) = \omega_i +\dfrac{1}{N+1} \sum_{j=1}^{N+1} A_{ij} (t) \sin(\theta_j (t) - \theta_i (t)),
\end{equation}
$i = 1,\hdots, N+1$, where the scalar $\theta_i (t)$ is the phase angle of oscillator $i$ at time $t$, the scalar $\omega_i$ is the frequency of  oscillator $i$. 
We require different frequencies for different oscillators, i.e., $\omega_i \neq \omega_j$, $\forall i \neq j$.
The time-varying connectivity is described by the adjacency matrix $A(t) = [A_{ij}(t)]$ where if node $i$ receives a connection from node $j$ at time $t$, then $A_{ij}(t) \neq 0$, otherwise $A_{ij}(t) = 0$.

We assume that a unidirectional adversarial connection from the intruder to one of the network nodes is  
established at some time $T > 0$. 
More precisely, the time-varying adjacency matrix is given by,
\begin{equation} \label{eq:Akuramoto}
    A(t) = \begin{cases}
        \begin{bmatrix}
            A_0 & \pmb{0} \\
            \pmb{0}^\top & 0
        \end{bmatrix}, \quad & t < T, \\
        \\
        \begin{bmatrix}
            A_0 & \bb^i \\
            \pmb{0}^\top & 0
        \end{bmatrix}, \quad & t \geq T.
    \end{cases}
\end{equation}
Here, $\bb^i$ is the vector that characterizes the attack from the intruder to node $i$ and $A_0$ is the binary $N$-dimensional  adjacency matrix describing the network connectivity, apart from the intruder.
In what follows, we use both the norm of the vector of phase perturbations and the order parameter (both defined next) to
characterize the effect of the attack on the nonlinear dynamics of coupled Kuramoto oscillators \eqref{eq:Kuramoto}.

The norm of the vector of phase perturbations is equal to,
\begin{equation}
    \| \delta \theta(t) \| := \sqrt{\sum_{j=1 }^{N}\left(\theta_j (t) - \frac{1}{N} \sum_{k=1}^{N} \theta_k (t)  \right)^2}. 
\end{equation}
It approaches a constant value in time if oscillators $i=1,\hdots, N$ are phased-locked and varies in time otherwise.
Also, we calculate the order parameter defined as,
\begin{equation} \label{eq:order}
{
    \rho(t) = \left| \dfrac{1}{N} \sum_{j=1}^{N} e^{i \theta_j(t)} \right|
}
\end{equation}
where $i = \sqrt{-1}$. The order parameter provides a normalized index of synchronization among all the oscillators: it is $\rho \approx 1$, when the oscillators are phase-locked, and $\rho \approx 0$ when there is no synchrony.

In our first example, we set  $T=15s$, and randomly select the frequencies $\omega_i$ from a normal distribution with a mean $0.5$ and a standard deviation $0.01$.
The initial conditions $\theta_j(0), \ j=1,\hdots, N$, are drawn from a standard normal distribution and then normalized such that $\sum_{j=1}^{N} \theta_j(0)^2 = 1$.
We set the initial condition of the intruder to be $\theta_{N+1} (0) = 0$.
The network topology and the attacks are given by, 
\begin{equation} \label{eq:bKuramoto}
    A_0 = \begin{bmatrix}
    0 & 1 & 0 \\ 
    1 & 0 & 1 \\ 
    1 & 1 & 0 
    \end{bmatrix}, \quad \bb^1 = \begin{bmatrix}
    -1 \\ 
    0 \\ 
    0 
    \end{bmatrix}, \quad \bb^2 = \begin{bmatrix}
    0 \\ 
    -1 \\ 
    0 
    \end{bmatrix}.
\end{equation}
We compare the two cases of an attack on node $1$ with indegree one and node $2$ with indegree two.
The former (latter) case corresponds to using $\bb^1$ ($\bb^2$) in Eq.\,\eqref{eq:Akuramoto}.
We repeat this numerical experiment 500 times and evaluate both $ \| \delta \theta(t) \|$ and $\rho$.

Figure \ref{fig:Kuramotorandom} shows that before the attack the oscillators become phase-locked and that after the attack phase locking is lost. 
We see a greater divergence rate both in the norm of the phase perturbation vector and in the order parameter, see panels \ref{fig:Kuramotorandom} A and B, respectively, when node 1 is attacked, compared to the case when node 2 is attacked.
This indicates that attacks targeting low-degree nodes result in a higher rate of transient instability.

\begin{figure}
    \centering
    \includegraphics[width=0.7\linewidth]{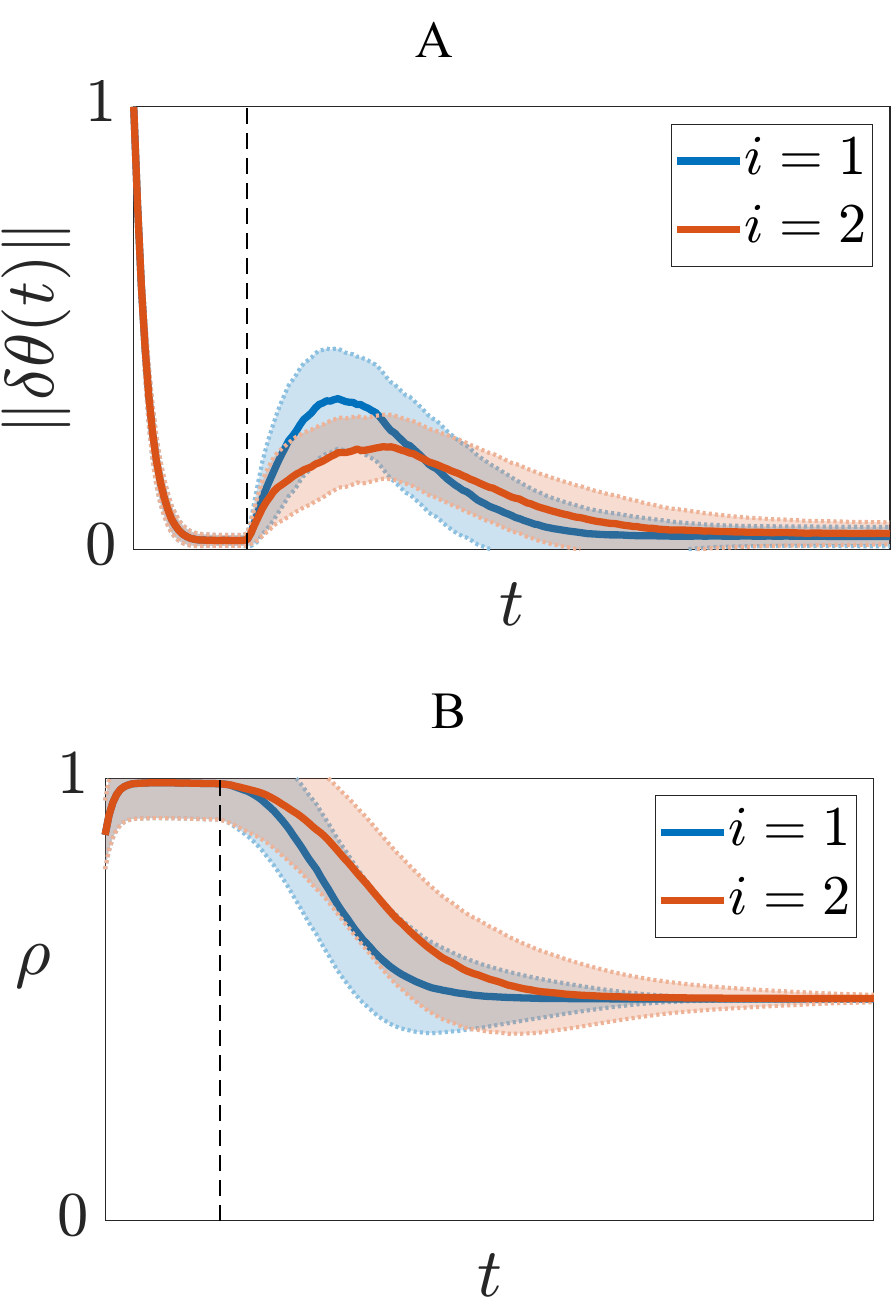}
    \caption{\textbf{Effect of the attacked node on the dynamics of a network of Kuramoto oscillators.}
We consider a network of Kuramoto oscillators and plot the norm of the phase perturbation vector, $\|\delta\theta(t)\|$ (panel A), and the order parameter, $\rho(t)$ (panel B), as functions of time $t$. The blue curves correspond to an attack on node $i=1$, i.e., $\mathbf{b}^{1}$ is selected in Eq.~\eqref{eq:bKuramoto}, whereas the red curves correspond to an attack on node $i=2$, i.e., $\mathbf{b}^{2}$ is selected. The dashed black line marks the onset of the attack at $t=T=15\,\mathrm{s}$. Prior to the attack, the blue and red curves are identical since the same initial conditions are used in both simulations. The results show the averages of $\|\delta\theta(t)\|$ (panel A) and $\rho(t)$ (panel B) over 500 realizations, each corresponding to a different random choice of the initial conditions and the natural frequencies of the oscillators, as described in the text. The shaded regions indicate one standard deviation across the realizations.}
    \label{fig:Kuramotorandom}
\end{figure}

We repeat the experiment above for a randomly generated scale-free network of coupled Kuramoto oscillators, which we choose to be the network of Fig. \ref{fig:synch_node} (reproduced for clarity in \ref{fig:Kuramoto_SF} A.)
Panels B and C show that before the attack the oscillators become phase-locked and that after the attack phase locking is lost. When  nodes of lower indegree are attacked,
we see a greater divergence rate both in the norm of the phase perturbation vector and in the order parameter, see panels \ref{fig:Kuramoto_SF} A and B, respectively. 
This confirms once again that nodes with lower indegrees (e.g., node 100) are more vulnerable to intruder attacks.

{A detailed asymptotic stability analysis for the case of Kuramoto dynamics is provided in the Methods Sec.\ \ref{Sec:MAIIII}.}

\begin{figure*}
    \centering
    \includegraphics[width = \linewidth]{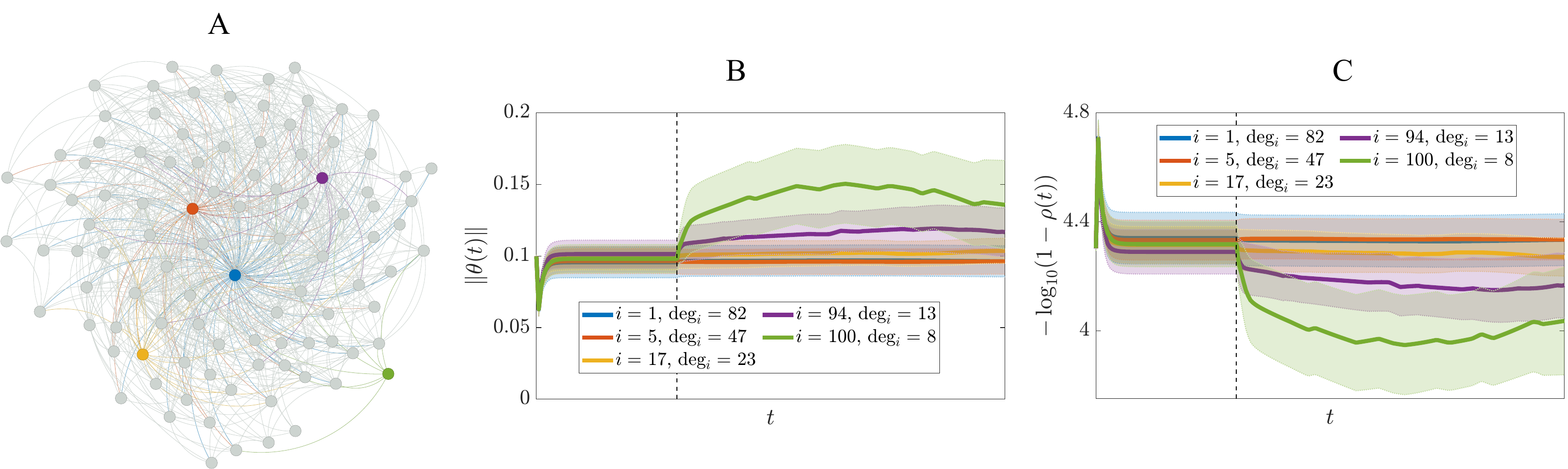}
    \caption{\textbf{Effect of the attacked node on destabilization of Kuramoto dynamics.}
    Panel A shows a randomly generated scale-free network under attack through either one of the colored nodes with different indegrees (blue $i=1$ with degree $82$, orange $i=5$ with degree $47$, yellow $i=17$ with degree $23$, purple $i=94$ with degree $12$, and green $i=100$ with degree $8$). 
    In panel B, we integrate the Kuramoto dynamics of Eq.\,\eqref{eq:Akuramoto} on this network, and plot the norm of the phase perturbations vector $ \| \delta \theta(t) \|$ as a function of time $t$ and for different choices of the attacked node $i$. 
    In panel C, we plot the order parameter $ \rho(t)$, Eq.\,\eqref{eq:order}, as a function of time $t$ and for different choices of the attacked node $i$. Different curves correspond to different attacked nodes, all with the same budget $c = 1$. 
    The dashed black line denotes the start time of the attack. 
    The initial conditions of the oscillators are chosen randomly from the uniform distribution $[0, \ 0.1]$. 
    The plot also provides the degree of node $i$ as $\deg_i$.
    The highlighted background denotes the standard deviation over 20 realizations.
    The natural frequencies of the oscillators are randomly drawn from a Normal distribution with a mean of 0.1 and a standard deviation of 0.001.
    }

    \label{fig:Kuramoto_SF}
\end{figure*}

\begin{remark}
    Our empirical results indicate that in the case of more than one intruder, the optimal attack, i.e., the attack that minimizes the algebraic connectivity, is the one with all the budget concentrated in one link from one of the intruders toward the node with the lowest indegree and with all the other intruders that remain disconnected.
    The particular choice of the intruder makes no difference, since all are identical before the attack.
    An extensive discussion on this point is provided in the Supplementary Note 2.
\end{remark}
\begin{remark}
    We have performed a numerical study over our randomly generated scale-free network, Fig.\ref{fig:Kuramoto_SF}A.
    We have seen that the following values of the budget can be considered `small' and `large' in our asymptotic analysis:
    \begin{itemize}
        \item the small budget should be $c < 10^{-2}$ (independent of the attacked node), 
        \item the large budget should be $c > 0.1 L_{ii}^2$ (dependent on the attacked node),
    \end{itemize}
    where $L_{ii}$ is the indegree of the attacked node.
    The details of our analysis are available in Supplementary Note 8.
\end{remark}

In Supplementary Note 9, we compare the effects of an intruder attack on two leaf nodes (nodes with degree 1): one connected to the hub node (the node with highest degree) and one not connected to the hub. By examining the dynamics immediately after the attack, we do not observe any difference between the two cases in their transient response.
\color{black}

\section{Discussion} \label{sec:conclusions}

 This work identifies a fundamental vulnerability of complex dynamical networks: a single node connected through adversarial links to one or more other nodes is sufficient to destabilize the network dynamics, regardless of the network size. Such an attack provides a generic mechanism for disrupting consensus in opinion dynamics, synchronization in oscillator networks, and desired formations in coupled autonomous systems. Furthermore, we find that targeting a node with the smallest indegree, {defined as the sum of the weights of all incoming connections}, typically destabilizes the network in the shortest time. In contrast, attacking hub nodes generally results in a slower onset of instability. Intuitively, this can be explained as follows: hubs respond less to the coupling with the intruder due to a large number of in-neighbors, while low-indegree nodes respond more.

The practical implementation of negative (adversarial) interactions depends on the application domain, but is in general well grounded across a wide range of systems. 
In technological and engineering networks, the distinction between cooperative and adversarial interactions often reduces to the sign of a feedback gain. In standard Laplacian coupling, the interaction acting on node $i$ takes the form $g (x_{\mathrm{ref}} - x_i)$, where $x_{\mathrm{ref}}$ is a reference state, $x_i$ is the state of node $i$, and $g$ is a scalar gain. Cooperative interactions correspond to $g > 0$, whereas adversarial interactions correspond to $g < 0$. Thus, changing the sign of the feedback gain naturally switches between the two regimes. In practice, such sign changes may arise from faulty or malicious agents, incorrect controller design, or cyber-attacks that invert transmitted signals. In power grid applications, for example, adversarial coupling can be realized through devices enforcing anti-phase behavior, as discussed in Subsection II.C.
In biological systems, negative interactions are ubiquitous and arise intrinsically. Examples include inhibitory synaptic coupling in neuronal networks, competitive interactions in ecological systems, and repression mechanisms in gene regulatory networks. In these settings, adversarial interactions are not externally imposed, but are a fundamental component of the system dynamics. Within this perspective, our framework captures the effect of introducing or amplifying inhibitory interactions in otherwise cooperative networks.
In social systems, negative coupling emerges through antagonistic or contrarian behavior. For instance, in opinion dynamics models, some individuals may systematically oppose the state of their neighbors (so-called ``contrarians''), effectively inducing negative interactions. Likewise, misinformation or adversarial agents in online networks can disrupt consensus by promoting opposing views. These mechanisms provide realistic pathways through which adversarial interactions can arise and influence collective dynamics.

 Our results are in stark contrast to the expectation that hubs are the most critical nodes of a network \cite{artime2024robustness}, and thus the ones that need to be more closely protected from attacks. 
 {This assumes particular relevance in the case of scale-free networks \cite{barabasi1999emergence}.
 A central dogma in the field of complex networks is that scale-free networks are stable against random failures (removal) of nodes but are vulnerable, in terms of structural integrity (i.e., the size of the strongly connected component) to attacks targeting the hubs. 
We show a significantly different picture for the case of dynamical attacks considered here: while the structure of scale-free networks is more vulnerable to attacks targeting the hubs \cite{artime2024robustness}, their dynamics is more vulnerable to attacks targeting nodes with low-indegree.
In summary, while hubs exhibit high vulnerability in terms of structural integrity, they are somewhat protected in dynamical contexts, leading to a situation where low-degree nodes are most vulnerable. Reference \cite{liu2011controllability} had also shown that the structural controllability of a network was enhanced when driver nodes are not connected to high-degree nodes. {Reference  \cite{menichetti2014network} had found that the structural controllability of a network depends strongly on the fraction of
low in-degree and low out-degree nodes. } However, destabilizing the network dynamics does not require controllability, and here we provide simple sufficient conditions to induce instability of this dynamics. Our work complements a large literature on the importance of low-degree nodes and weak ties \cite{granovetter1973strength}; the relevance of low-degree nodes, under specific circumstances, is also acknowledged in the literature on percolation dynamics on networks  \cite{artime2024robustness} and {the importance of weak ties on stability for predator–prey networks has been pointed out in \cite{allesina2012stability}.} 
 }

 Our work complements the large literature that focuses on conditions to ensure the stability of the network dynamics, see e.g., \cite{boccaletti2006complex}, while we focus on instability and on characterizing the severity of such instability. We investigate the fundamental question of how desired dynamical states can be disrupted by the purposeful insertion of one or a few intruder agents. 
Our intruders are structurally similar to the other network agents, which is consistent with the case of cyber-attacks. 
Our results are initially derived for the case of the linear consensus problem but have implications in a broad variety of network dynamics, including opinion dynamics, synchronization, formation control, and power grid balancing. We also study the case of the Kuramoto model, 
which is relevant to synchronization of biological systems. That having been said, we do not claim that our results apply to all the possible realizations of Eq.\ \eqref{general}, but to a variety of cases of interest.

 Our work on intruder attacks in networked systems intersects with the domain of network interdiction, which has been extensively studied in operations research and optimization contexts. Smith and Song's comprehensive survey on network interdiction \cite{smith2020survey} provides a systematic overview of how adversaries can strategically disrupt network operations by targeting critical components. While traditional network interdiction focuses primarily on structural disruptions such as removing nodes or edges to maximize the disruption of flow, connectivity, or shortest paths \cite{Wood1993,IsraeliWood2002}, our approach examines dynamical disruptions through adversarial connections that destabilize the network's equilibrium states. This distinction is crucial: classical interdiction models often assume a static network structure where the adversary's goal is to disconnect components or increase path lengths \cite{SmithSong2020,CappaneraScaparra2011}, whereas our framework considers how a single intruder node can exploit the network's dynamical properties to induce instability while leaving the structural connectivity intact. 
 Our results complement the interdiction literature by demonstrating that in cyber-physical systems, power grids, and consensus networks, dynamical vulnerability differs fundamentally from structural vulnerability: targeting low-indegree nodes rather than high-degree hubs identified as critical in classical interdiction models \cite{HolmeKimYoonHan2002} produces the most severe instabilities. This finding suggests that protection strategies derived from structural interdiction theory may be insufficient for dynamical networks, necessitating new defense mechanisms that account for both topological and dynamical characteristics.

We show the emergence of universal scaling properties, with small perturbations causing uniform responses and larger ones exposing specific vulnerabilities, especially at the low-indegree nodes.
When applied to real complex networks, we see that it may not always be easy to characterize the rate of instability produced by the attack in the general case of non-balanced digraphs. 
However, our theory (see also Fig.\,\ref{fig:realnet}) shows that the rate of instability generated by an attack on a single node averaged over all $N$ networks nodes, scales exactly as $-N^{-1}$,  for the case that the adversarial connection is unidirectional. This indicates that larger networks are more robust, on average, against intruder attacks.

Although our work focuses on the case that an intruder node is added to an existing network, it can be generalized to the case that an existing network node becomes an intruder. In that case, the strategy that maximizes the effects of an attack is for the intruder to target only one node with the lowest indegree, where the indegree is computed by counting all the nodes in the network, except for the intruder. 

An obvious countermeasure against intruder attacks is to disconnect the attacker(s) from the network. This is consistent, for example, with the load shedding strategy, commonly implemented for power grids in the case of a power imbalance. 
In general,   shedding the part of the network that includes the intruder provides a viable solution against intruder attacks, which is easier to do when the attacker is connected to the rest of the network by a single link compared to the case of multiple links. This observation highlights a trade-off between the severity of an attack and the recoverability from such attack.

{Extending the present framework to higher-order (many-body) interactions represents a natural and promising direction for future work. Recent studies have shown that such interactions can significantly modify synchronization properties, leading to global topological synchronization in some cases, but also to multistability and abrupt desynchronization in others \cite{carletti2023global, gallo2022synchronization, skardal2019abrupt, millan2020explosive}. This indicates that, depending on their structure, higher-order interactions may either promote coherence or give rise to additional instability mechanisms. In this context, we expect them to shift stability thresholds and modulate the system’s sensitivity to perturbations, without altering the underlying mechanism.}

Finally, a word of caution. While our results are rigorous for the case of linear dynamics (e.g., consensus), they are mostly based on extensive numerical calculations in the case of nonlinear dynamics.  We believe more work is needed to fully characterize the effects of intruder attacks in nonlinear networks, although our work provides strong indication that in many cases of interest, the rate of instability observed for linear networks is inherited by their nonlinear counterparts.

\section{Methods} \label{sec:methods}

\subsection{Transverse reactivity of stable and unstable consensus dynamics}

The first part of this subsection discusses general linear dynamical systems.
The second part of this subsection focuses on linear consensus dynamics as a special case of general linear dynamical systems. 
Consider a linear dynamical system
\begin{equation} \label{eq:lti}
    \dot{\bx} (t) = M \bx(t)
\end{equation}
where $\bx(t) \in \mathbb{R}^N$ is the vector of system states and the square matrix $M \in \mathbb{R}^{N \times N}$.
\begin{definition}{Spectral abscissa of a matrix.} Given a square matrix $M \in R^{N \times N}$, we define the spectral abscissa as the greatest real part of the matrix spectrum,
\begin{equation}
    \alpha(M)=\max_i \Re(\lambda_i),
\end{equation}
where $\lambda_1,\lambda_2,...,\lambda_N$ are the possibly complex eigenvalue of the matrix $M$ and $\Re(\cdot)$ returns the real part of its argument. 
\end{definition}
The origin is a stable fixed point for the system in Eq.\,\eqref{eq:lti} if and only if $\alpha(M) < 0$.

\begin{definition}{Reactivity of a linear dynamical system \cite{Neubert1997ALTERNATIVES}.} The reactivity for the system in Eq.\,\eqref{eq:lti} is defined as,
\begin{align} 
\label{eq:reactivity}
\begin{split}
    \xi\left( M \right) & := \max_{\|  \bX\| \neq 0} \left[ \frac{1}{\|  \bX \|} \frac{d \|  \bX\|}{d t} \right] \\
    & = \max_{  \bX \neq 0} \frac{{\bX}^\top \left( \frac{ M+M^\top}{2} \right) \bX }{{ \bX}^\top  \bX} \\
    & =  \lambda_{\max} \left( \frac{M+M^\top}{2} \right) =:  \lambda_{\max}(S),
\end{split}
\end{align}
where the matrix $S = (M+M^\top)/2$ is the symmetric part of the matrix $M$ and $\lambda_{\max}(\cdot)$ indicates the largest eigenvalue of the symmetric matrix in its argument.
\end{definition}
The reactivity measures the maximal rate of instantaneous growth/decay of the norm of the state vector. 
If $\xi\left( M \right) < 0$ ($\xi\left( M \right) >0$) the norm of the state vector will decay (may grow) instantaneously.
The choice of $\bX$ that maximizes \eqref{eq:reactivity} is the eigenvector $\bv_1$ of the matrix $S$ associated with the eigenvalue $\lambda_{\max}(S)$.
Note that the reactivity is computed by maximizing over $\bX$, thus it provides a worst-case instantaneous rate of growth for $\| \bX \|$. In the case that $\bX$ is generically chosen, it will have a non-zero component along the eigenvector $\bv_1$, thus the worst-case rate of growth for $\| \bX \|$ will be observed with probability 1.

\begin{theorem} \label{theo}
    The reactivity of the dynamics \eqref{eq:lti} is greater or equal to the spectral abscissa of the matrix $M$, i.e., $\xi(M) \geq \alpha(M)$.
\end{theorem}
\begin{proof}
We set $S = (M+M^\top)/2$ and write the reactivity of the dynamics \eqref{eq:lti},
\begin{align} \label{eq:lambda}
\begin{split}
    \xi(M) = \max_{ \bZ \neq 0} \frac{{\bZ}^\top S \bZ }{{ \bZ}^\top \bZ} & \geq { \bv_1}^\top S \bv_1 \\
    & =   { \bv_1}^\top \left(\frac{M+M^\top}{2} \right) \bv_1 \\
    & =  \frac{1}{2}(\bv_1^\top M \bv_1 + \bv_1^\top M^\top \bv_1 ) \\
    & = \frac{1}{2}(\bv_1^\top \alpha(M) \bv_1  + \alpha(M) \bv_1^\top \bv_1 ) \\
    & = \alpha(M) 
\end{split}
\end{align}
This concludes the proof. 
\end{proof}

Theorem \ref{theo} is important as it states that the instantaneous rate of growth of the state $\bZ$ in \eqref{eq:lti} can be either larger than or equal to the asymptotic rate of growth of $\bZ$. This is true in both the cases of asymptotically stable and unstable dynamics. 
If $\alpha(M)>0$, we then see that $\xi(M) \geq \alpha(M) >0$. This indicates that for an unstable system, the rate of instantaneous growth is positive and exceeds or is equal to the rate of asymptotic growth.

It is inferred that the inequality in \eqref{eq:lambda} is satisfied with the equal sign when the left and right eigenvectors of the matrix $M$ associated with the eigenvalue $\lambda_1$ coincide. 
This means that $\bZ^* =  \bv_1$ becomes the maximizer in \eqref{eq:lambda}.

The discussion that follows is for the single-integrator consensus dynamics in Eq.\,\eqref{consensus}, $\dot{\bX}(t) = -L \bX(t)$, which is a special case of the general linear dynamical system in Eq.\,\eqref{eq:lti}.

We write the eigenvalue equation $L \bv_i=\lambda_i \bv_i$, where $\lambda_i$  ($\bv_i$) are the possibly complex eigenvalues (eigenvectors) of the Laplacian matrix $L$, $i=1,...,N$. By the property that the rows of the Laplacian all sum to zero, there exists one eigenvalue of the Laplacian equal to zero, with the corresponding eigenvector having entries that are all the same. We index this eigenvalue $j$.
 Without loss of generality, we set $\Re(\lambda_1(L)) \leq \Re(\lambda_2(L)) \leq ... \leq \Re(\lambda_N(L))$,
where $\Re(\cdot)$ returns the real part of its argument. When the Laplacian matrix is proper, $j=1$ and $\Re(\lambda_1(L))=0$, which results in an asymptotically stable consensus dynamics Eq.\,\eqref{consensus}. Here, we remove the assumption that the Laplacian is proper and present a general theory of transient and asymptotic stability for Eq.\,\eqref{consensus}.


Any point such that
\begin{equation} \label{cm}
x_1=x_2=...=x_N
\end{equation}
is an equilibrium for the dynamics \eqref{consensus}. The set of points that satisfy \eqref{cm} define the consensus manifold. We are interested in whether the consensus manifold is asymptotically stable or unstable and in characterizing the maximum rate of transient instability away from this manifold. Thus in what follows, we focus on the transverse consensus dynamics,
\begin{equation}
\label{transverse}
    \dot{\bZ}(t) = -A \bZ(t)
\end{equation}
$\bZ \in \mathbb{R}^{N-1}=V^\top \bX$, the matrix $A \in \mathbb{R}^{N-1 \times N-1} =V^\top L V$
and $V \in \mathbb{R}^{N \times N-1}$ is an orthonormal basis for the null subspace of $\bone^\top$, i.e., $V$ is a matrix whose columns are normal and orthogonal to one another and have zero column sums.
\begin{remark}
The matrix $A$ has the same spectrum as the Laplacian matrix $L$, except for one eigenvalue of the matrix $L$ which is equal zero, i.e., $\mathcal{S}(L) = \mathcal{S}(A) \cup \{ 0 \}$ where $\mathcal{S}(\cdot)$ returns the spectrum of the matrix in its argument.
\end{remark}

In our previous work \cite{nazerian2023single}, we studied the entire consensus dynamics Eq.\,\eqref{consensus} under the assumption the consensus dynamics was stable, i.e., $\alpha(-A) < 0$.
In this case, the trajectories of the system states eventually approach the origin (asymptotic stability), while the magnitude of the states may grow at transient times due to positive reactivity, which is not desirable in various situations, such as linearized dynamics about a fixed point. This transient growth may result in the invalidity of the linearization assumption.
Unlike previous work by us and others \cite{Farrell1996Generalized,Neubert1997ALTERNATIVES,Hennequin2012Nonnormal,Tang2014Reactivity,Biancalani2017Giant,Asllani2018Structure,MUOLO2019Patterns,Gudowska2020From,Lindmark2021Centrality,Duan2022Network,nazerian2023single,Nazerian2023Reactability},  here we are mostly concerned with the reactivity of the transverse dynamics, i.e., Eq.\,\eqref{transverse} in the case that $\alpha(-A) > 0$, i.e., the system Eq.\,\eqref{consensus} is asymptotically unstable.

Next, we comment on the relevance of both the spectral abscissa and the reactivity on the transverse consensus dynamics \eqref{transverse}. The asymptotic rate of $\bZ(t)$ in \eqref{transverse} is given by $\alpha(-A)$, i.e., the spectral abscissa of $-A$.
In general, $\alpha(-A)$ can be either positive or negative. 
On the other hand, the instantaneous rate of growth of the norm of the vector $\bZ$ is given by the reactivity,
\begin{align} \label{eq:reactivity0}
\begin{split}
    \xi\left( -A \right) & =  \max_{\|  \bZ\| \neq 0} \left[ \frac{1}{\|  \bZ \|} \frac{d \|  \bZ\|}{d t} \right] = \lambda_{\max} \left( -\frac{A+A^\top}{2} \right).
\end{split}
\end{align}
Thus a positive (negative) reactivity $\xi(-A)$ indicates that the norm of the state $\bZ$ tends to grow (shrink) in the limit of $t \rightarrow 0$. 

\begin{proposition}
    The algebraic connectivity of the graph represented by the Laplacian matrix in Eq.\,\eqref{consensus} is equal to the negative of the reactivity of the transverse dynamics in Eq.\,\eqref{transverse}, that is, $f(L) \equiv - \xi(-A)$.
\end{proposition}
\begin{proof}
    The reactivity of the dynamics in Eq.\,\eqref{eq:reactivity0} is
    \begin{align*}
        \xi\left( -A \right) & :=  \max_{\|  \bZ\| \neq 0} \left[ \frac{1}{\|  \bZ \|} \frac{d \|  \bZ\|}{d t} \right] \\
        & = \max_{\|  \bZ\| \neq 0} \left[ \frac{1}{\sqrt{\bZ^\top \bZ}} \frac{d \sqrt{\bZ^\top \bZ}}{d t} \right] \\
        & = \max_{\|  \bZ\| \neq 0} \left[ \frac{1}{\sqrt{\bZ^\top \bZ}} \frac{\dot{\bZ}^\top \bZ + \bZ^\top \dot{\bZ}}{2\sqrt{\bZ^\top \bZ}} \right] \\
        & = \max_{\|  \bZ\| \neq 0} \left[ \frac{\bZ^\top (-A - A^\top) \bZ}{2 \bZ^\top \bZ} \right].
    \end{align*}
    Also, $V^\top V = I$, $L = V A V^\top$ and $\bX = V \bZ$ which indicates $\bX \perp \bone$ since $\bX \in \text{range}(V)$, so
    {\footnotesize
    \begin{align*}
        \xi\left( -A \right) & = \max_{\|  \bZ\| \neq 0} \left[ \frac{\bZ^\top  \left(-(V^\top V) A (V^\top V) - (V^\top V) A^\top (V^\top V)\right) \bZ}{2 \bZ^\top \bZ} \right] \\
        & = \max_{\|  \bZ\| \neq 0} \left[ \frac{(\bZ^\top V^\top)  (- L - L^\top) (V\bZ)}{2 (\bZ^\top V^\top) (V\bZ)} \right] \\
        & = \max_{ \bX \neq 0, \bX \perp \bone} \left[ \frac{\bX^\top  (- L - L^\top) \bX}{2 \bX^\top \bX} \right] \\
        & = -\min_{ \bX \neq 0, \bX \perp \bone} \left[ \frac{\bX^\top  ( L + L^\top) \bX}{2 \bX^\top \bX} \right] \\
        & = -\min_{ \bX \neq 0, \bX \perp \bone} \left[ \frac{\bX^\top  L \bX}{\bX^\top \bX} \right] \\
        & = -f (L).
    \end{align*}}
    That concludes the proof.
\end{proof}

In what follows, we refer to $\xi(-A)$ in Eq.\ \eqref{eq:reactivity0} as either the reactivity of the dynamics in Eq.\,\eqref{transverse} or the transverse reactivity of the consensus dynamics in Eq.\,\eqref{consensus}, interchangeably.

\begin{corollary}
    The instantaneous rate of growth of the transverse consensus dynamics in Eq.\,\eqref{transverse} is greater or equal to the asymptotic rate of growth, i.e., $\xi (-A) \geq \alpha (-A)$.
\end{corollary}
\begin{proof}
    The relation $\xi (-A) \geq \alpha (-A)$ is obtained by placing $M = -A$ in Theorem \ref{theo}.
\end{proof}

To illustrate the relation $\xi (-A) \geq \alpha (-A)$, we provide a numerical example of stable and unstable dynamics.
First, we randomly generate the proper Laplacian matrix
\begin{equation} \label{Lproper}
    L = \begin{bmatrix}
    1 & -1 & 0 & 0 & 0 \\ 
    0 & 1 & 0 & 0 & -1 \\ 
    -1 & -1 & 2 & 0 & 0 \\ 
    0 & 0 & 0 & 0 & 0 \\ 
    0 & 0 & 0 & -1 & 1 
    \end{bmatrix}
\end{equation}
with all eigenvalues with non-negative real-part and only one zero eigenvalue.
The matrix $A$ in Eq.\,\eqref{transverse} corresponding to the Laplacian matrix $L$ in \eqref{Lproper} is
\begin{equation}
    A = \begin{bmatrix}
    1.5427 & -0.1809 & 0.2337 & -0.9045 \\ 
    -0.1483 & 2.1281 & 0.5427 & 0.4045 \\ 
    0.5427 & -0.1809 & 0.2337 & 0.0955 \\ 
    0.5427 & -0.1809 & -0.7663 & 1.0955 
    \end{bmatrix}.
\end{equation}
Here, $\xi(-A) = -0.0250 \geq \alpha(-A) = -1$. 
We simulate Eq.\,\,\eqref{transverse} from the initial condition $\bZ(0) = [-0.2187,\ -0.1115,\ 0.9320 ,\ 0.2667 ]$ and plot the norm of the transverse perturbation $\| \bZ(t) \|$ in Fig.\,\ref{fig:consensus_rate} A the in log scale.
As evident by the slope of $\| \bZ(t) \|$ vs $t$, the transient rate of decay $\xi(-A)$ is larger than the asymptotic rate of decay $\alpha(-A)$.

Next, we add an adversarial agent to the above Laplacian matrix with a budget $-c = -0.001$ and obtain
\begin{equation} \label{Lnonproper}
    L_{aug} = \begin{bmatrix}
    1 & -1 & 0 & 0 & 0 & 0 \\ 
    0 & 1 & 0 & 0 & -1 & 0 \\ 
    -1 & -1 & 2 & 0 & 0 & 0 \\ 
    0 & 0 & 0 & -0.001 & 0 & 0.001 \\ 
    0 & 0 & 0 & -1 & 1 & 0 \\ 
    0 & 0 & 0 & 0.001 & 0 & -0.001 
    \end{bmatrix}
\end{equation}
which has mixed signed real-part eigenvalues.
For the Laplacian matrix \eqref{Lnonproper}, the corresponding matrix $A_{aug}$ in Eq.\,\eqref{transverse} is
\begin{equation}
    A_{aug} = \begin{bmatrix}
    1.4923 & -0.1527 & 0.2024 & -0.9160 & 0.0840 \\ 
    -0.2178 & 2.1372 & 0.4923 & 0.3739 & 0.3739 \\ 
    0.4923 & -0.1527 & 0.2014 & 0.0840 & 0.0850 \\ 
    0.4923 & -0.1527 & -0.7976 & 1.0840 & 0.0840 \\ 
    0.4923 & -0.1527 & 0.2034 & 0.0840 & 0.0830 
    \end{bmatrix}.
\end{equation}
Here, $\xi(-A_{aug}) = 0.0874 \geq \alpha (-A_{aug}) = 0.002$.
We simulate Eq.\,\,\eqref{transverse} from the initial condition $\bZ(0) = [0.1202,\  0.1937,\ 0.1545,\ 0.8538,\ 0.4418]$ and plot the norm of the transverse perturbation $\| \bZ(t) \|$ in Fig.\,\ref{fig:consensus_rate} B in the log scale.
As evident by the slope of $\| \bZ(t) \|$ vs $t$, the transient rate of decay $\xi(-A)$ is larger than the asymptotic rate of decay $\alpha(-A)$.

\begin{figure}
    \centering
    \includegraphics[width=0.6\linewidth]{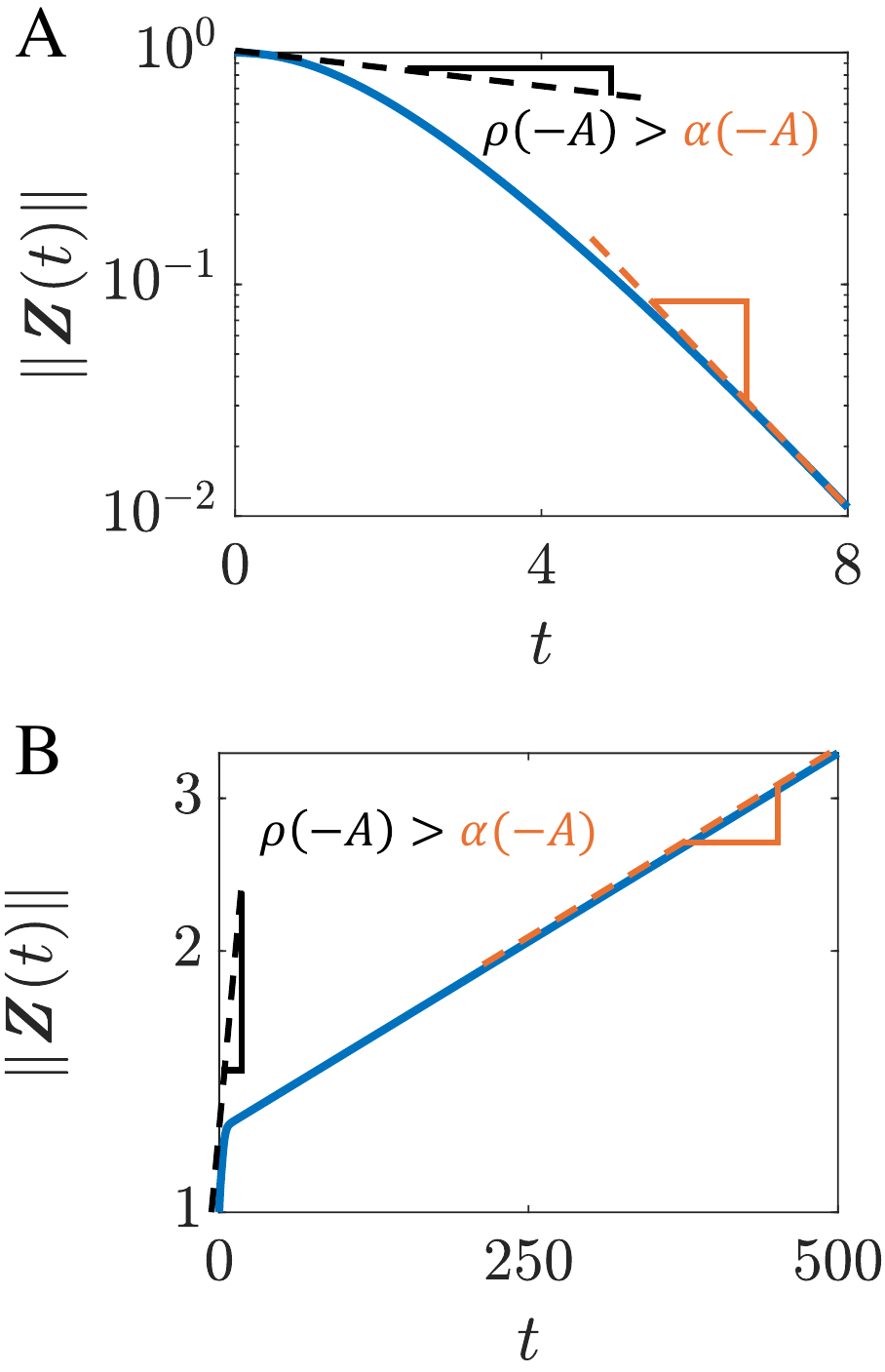}
    \caption{\textbf{Transient growth and decay in consensus dynamics.}
Panels A and B show the transverse consensus dynamics governed by Eq.~\eqref{transverse} for the cases of stable and unstable dynamics, respectively. In both panels, the vertical axis is shown on a logarithmic scale, so that the slopes of the curves represent the rates of decay (panel A) and growth (panel B). In both cases, the transient rate, $\xi(-A)$, exceeds the corresponding asymptotic rate, $\alpha(-A)$, demonstrating that the initial decay or growth is faster than its long-term asymptotic behavior.}
    \label{fig:consensus_rate}
\end{figure}

\subsection{Proof of Proposition \ref{prop1}} \label{sec:prop1}

Given the possibly asymmetric Laplacian $L$ and the budget $-c < 0$, we solve
        \begin{align}
        \begin{split}
            \min_{\bb, \bx } \quad 
            & \bx^\top \left(V^\top \dfrac{{L_{aug}^b}+{L_{aug}^b}^\top}{2} V \right)\bx \\
            \text{subject to} \quad 
            & \bx^\top \bx = 1 , \\
            & L_{aug}^b = \begin{bmatrix}
                L + \text{diag}(\bb) & -\bb \\
                -\bb^\top & -c
            \end{bmatrix}, \\
            & b_i \leq 0, \quad \forall i = 1, \hdots, N,\\
            & \sum_{i=1}^N b_i = -c.
        \end{split}
        \end{align}
    First, we set $\by:= V\bx$ which also implies $\bone^\top \by = 0$ and rewrite the optimization in $\by$ and $\bb$:
        \begin{align}
        \begin{split}
            \min_{\bb, \by } \quad 
            & \by^\top \dfrac{{L_{aug}^b}+{L_{aug}^b}^\top}{2} \by \\
            \text{subject to} \quad 
            & \by^\top \by = 1 , \\
            & \bone^\top \by = 0, \\
            & {L_{aug}^b} = \begin{bmatrix}
                L + \text{diag}(\bb) & -\bb \\
                -\bb^\top & -c
            \end{bmatrix}, \\
            & b_i \leq 0, \quad \forall i = 1, \hdots, N,\\
            & \sum_{i=1}^N b_i = -c.
        \end{split}
        \end{align}
In the above optimization, we set $\by = : [\by_r^\top, y_{N+1}]^\top$ and $\by_r^\top = [y_1, \ y_2, \hdots, y_N]^\top \in \mathbb{R}^N$, i.e., we separate the last entry of the vector $\by$ from the remaining entries.
We also assume $y_i \neq y_j, \forall i\neq j, i= 2, \hdots, N, j = 2, \hdots, N.$
The objective function is
\begin{align}
\begin{split}
    \by^\top \dfrac{{L_{aug}^b}+{L_{aug}^b}^\top}{2} \by = & -y_{N+1}^2 c+ \by_r^\top \frac{L+L^\top}{2} \by_r \\ 
    &  -2y_{N+1} \by_r^\top \bb +\by_r^\top \diag(\bb) \by_r.
\end{split}
\end{align}
The two terms $-2y_{N+1} \by_r^\top \bb +\by_r^\top \diag(\bb) \by_r$ in the above equation that explicitly depend on $\bb$ are written as
\begin{align}
\begin{split}
    -2y_{N+1} \by_r^\top \bb + \by_r^\top \diag(\bb) \by_r & = \sum_{i=1}^N -2 y_{N+1} y_i b_i  + \sum_{i=1}^N y_i^2 b_i \\
    & = \sum_{i=1}^N \left( y_i^2 -2 y_{N+1} y_i  \right) b_i \\
    & =: \sum_{i=1}^N \beta_i b_i
\end{split}
\end{align}
where $\beta_i =  y_i^2 -2 y_{N+1} y_i, \ i = 1, 2, \hdots, N$.
Since these two terms in the objective function depend linearly on $\bb$ and $\sum_j b_j = -c$ and $b_j \leq 0, \forall j$, thus one focuses all budget on the node with the maximum $\beta_i$, i.e., $b_{i^*} = -c$, $b_j = 0, \ \forall j \neq i^*$, and $i^* = \arg \max_i \beta_i$. 
That concludes the proof. \hfill \qed

{We remark that Proposition 1 simply predicts under generic assumptions the particular structure of $b^*$, the optimal $b$, namely that all the entries of $b^*$ will be zero except for one, but it does not predict which one.}

\subsection{Proof of Proposition \ref{prop2}} \label{sec:prop2}

We evaluate the algebraic connectivity using \eqref{eq:algeb} and take the vector $\bX_0 = [N, \ -1, \allowbreak \ -1, \allowbreak \cdots, \allowbreak  -1]^\top$ such that its entries sum to 0.
Evaluating the Rayleigh quotient for $L_{aug}$ from \eqref{eq:Lnewbi} and $\bX_0$ yields
\begin{equation*}
    f \leq \dfrac{\bX_0^\top L_{aug} \bX_0}{\bX_0^\top \bX_0} = -\dfrac{(N+1)^2 c}{N^2 + N} \leq 0.
\end{equation*}
Thus $f \leq 0$. \hfill \qed

\subsection{Proof of Proposition \ref{prop3}} \label{sec:prop3}

Given the possibly asymmetric Laplacian $L$ and the budget $-c < 0$, we solve
        \begin{align}
        \begin{split}
            \min_{\bb, \bx } \quad 
            & \bx^\top \left(V^\top \dfrac{{L_{aug}^u}+{L_{aug}^u}^\top}{2} V \right)\bx \\
            \text{subject to} \quad 
            & \bx^\top \bx = 1 , \\
            & L_{aug}^u = \begin{bmatrix}
                L + \text{diag}(\bb) & -\bb \\
                \pmb{0}^\top & 0
            \end{bmatrix}, \\
            & b_i \leq 0, \quad \forall i = 1, \hdots, N,\\
            & \sum_{i=1}^N b_i = -c.
        \end{split}
        \end{align}
    First, we set $\by:= V\bx$ which also implies $\bone^\top \by = 0$ and rewrite the optimization in $\by$ and $\bb$:
        \begin{align}
        \begin{split}
            \min_{\bb, \by } \quad 
            & \by^\top \dfrac{{L_{aug}^u}+{L_{aug}^u}^\top}{2} \by \\
            \text{subject to} \quad 
            & \by^\top \by = 1 , \\
            & \bone^\top \by = 0, \\
            & L_{aug}^u = \begin{bmatrix}
                L + \text{diag}(\bb) & -\bb \\
                \pmb{0}^\top & 0
            \end{bmatrix}, \\
            & b_i \leq 0, \quad \forall i = 1, \hdots, N,\\
            & \sum_{i=1}^N b_i = -c.
        \end{split}
        \end{align}
In the above optimization, we set $\by = : [\by_r^\top, y_{N+1}]^\top$ and $\by_r^\top = [y_1, \ y_2, \hdots, y_N]^\top \in \mathbb{R}^N$, i.e., we separate the last entry of the vector $\by$ from the remaining entries.
We also assume $y_i \neq y_j, \forall i\neq j, i= 2, \hdots, N, j = 2, \hdots, N.$
The objective function is
\begin{align}
\begin{split}
    \by^\top \dfrac{{L_{aug}^u}+{L_{aug}^u}^\top}{2} \by = & -y_{N+1} \by_r^\top \bb + \by_r^\top \frac{L+L^\top}{2} \by_r \\ 
    & +\by_r^\top \diag(\bb) \by_r
\end{split}
\end{align}
The two terms $-y_{N+1} \by_r^\top \bb +\by_r^\top \diag(\bb) \by_r$ in the above equation that explicitly depend on $\bb$ are written as
\begin{align}
\begin{split}
    -y_{N+1} \by_r^\top \bb + \by_r^\top \diag(\bb) \by_r & = \sum_{i=1}^N - y_{N+1} y_i b_i  + \sum_{i=1}^N y_i^2 b_i \\
    & = \sum_{i=1}^N \left( y_i^2 - y_{N+1} y_i  \right) b_i \\
    & =: \sum_{i=1}^N \beta_i b_i
\end{split}
\end{align}
where $\beta_i =  y_i^2 - y_{N+1} y_i, \ i = 1, 2, \hdots, N$.
Since these two terms in the objective function depend linearly on $\bb$ and $\sum_j b_j = -c$ and $b_j \leq 0, \forall j$, thus one focuses all budget on the node with the maximum $\beta_i$, i.e., $b_{i^*} = - c$, $b_j = 0, \ \forall j \neq i^*$, and $i^* = \arg \max_i \beta_i$. 
That concludes the proof. \hfill \qed

\subsection{Proof of Proposition \ref{prop4}} \label{sec:prop4}

    We evaluate the algebraic connectivity using \eqref{eq:algeb} and take the vector $\bX_0 = [N, \ -1, \allowbreak \ -1, \allowbreak \cdots, \allowbreak  -1]^\top$ such that its entries sum to 0.
    Evaluating the Rayleigh quotient for $L_{aug}$ from \eqref{eq:Lnewuni} and $\bx_0$ yields
    \begin{equation*}
        f \leq \dfrac{\bX_0^\top L_{aug} \bX_0}{\bX_0^\top \bX_0} = -\dfrac{(N+1) c}{N^2 + N} \leq 0.
    \end{equation*}
    Thus $f \leq 0$, and the proof is complete. \hfill \qed

\subsection{Digraphs with a Unidirectional Connection} \label{sec:uni}

Here we use matrix perturbation theory \cite{bamieh2022tutorial} to predict the changes in the eigenvalues of a given matrix $A_0$ subject to a small perturbation $\epsilon A_1$, i.e., $\Lambda_\epsilon = \Lambda_0 + \epsilon \Lambda_1 + \epsilon^2 \Lambda_2 + \cdots$ where $\Lambda_\epsilon$ is the estimated eigenvalues of the matrix $A_\epsilon = A_0 + \epsilon A_1$, $\Lambda_0$ is the exact eigenvalues of $A_0$, and $\epsilon^k \Lambda_k, \ \forall k \geq 1$, are the $k$th order approximation of changes in the eigenvalues $\Lambda_0$.

In the following, we focus on the case of directed networks, i.e., for which the matrix $A_0$ is asymmetric and set $k = 1$.
Thus, $\Lambda_1 = \diag(W_0^* A_1 V_0)$, where $W_0$ and $V_0$ are the matrices of the left and the right eigenvectors of $A_0$ such that $W_0^* V_0 = I$, and the superscript ${}^*$ denotes conjugate transpose.

\textbf{The limit of small $c$ and unidirectional connection:} In this case, we set $\epsilon = c$ and
\begin{equation}
    A_0 =  V^\top \dfrac{\tilde{L} + \tilde{L}^\top}{2} V, \quad A_1 = V^\top \dfrac{P_i + P_i ^\top}{2} V
\end{equation}
 where 
\begin{align} \label{eq:lagematrix2}
    \tilde{L} = 
    \begin{bmatrix}
    & & & & & & & 0\\ 
    & & & & & & & \vdots \\ 
    & & & & & & & 0 \\ 
    & & & L & & & & 0 \\ 
    & & & & & & & 0 \\ 
    & & & & & & & \vdots \\ 
    & & & & & & & 0 \\ 
    0 & \hdots & 0 & 0 & 0 & \hdots & 0 & 0
    \end{bmatrix}, \\
    P_i = 
    \begin{bmatrix}
    0 & \hdots & 0 & 0 & 0 & \hdots & 0 & 0 \\ 
    \vdots & \ddots & \vdots & \vdots & \vdots & \ddots & \vdots & \vdots \\ 
    0 & \hdots & 0 & 0 & 0 & \hdots & 0 & 0 \\ 
    0 & \hdots & 0 & -1 & 0 & \hdots & 0 & 1 \\ 
    0 & \hdots & 0 & 0 & 0 & \hdots & 0 & 0 \\ 
    \vdots & \ddots & \vdots & \vdots & \vdots & \ddots & \vdots & \vdots \\ 
    0 & \hdots & 0 & 0 & 0 & \hdots & 0 & 0 \\ 
    0 & \hdots & 0 & 0 & 0 & \hdots & 0 & 0 
    \end{bmatrix}
\end{align}
where $P_i$ has the entry -1 on the main diagonal on the $i$th position.
The first order matrix perturbation theory is then
\begin{align}
\begin{split}
    f_b & = f_0 + c \dfrac{df}{db} \\
    & = \bv_0^\top A_0 \bv_0 + c \bv_0^\top A_1 \bv_0 \\
    & = \bv_0^\top V^\top \dfrac{\tilde{L} + \tilde{L}^\top}{2} V \bv_0 + c \bv_0^\top V^\top \dfrac{P_i + P_i ^\top}{2} V \bv_0 \\
    & =: \by_0^\top \dfrac{\tilde{L} + \tilde{L}^\top}{2} \by_0 + c \by_0^\top \dfrac{P_i + P_i ^\top}{2} \by_0 \\
    & =: f_0 + c \by_0^\top \dfrac{P_i + P_i ^\top}{2} \by_0
\end{split}
\end{align}
where $\by_0 := V \bv_0$ and $\by_0^\top\left(\frac{\tilde{L} + \tilde{L}^\top}{2} \right) \by_0 = f_0 $. 
Note that since the graph is balanced, the Laplacian matrix $L$ has both row and column sums equal to zero. 
Also, considering the block diagonal structure of the matrix $\tilde{L}$ results in having two zero eigenvalues, one corresponding to the eigenvector with all entries equal ones and the other eigenvalue $f_0 = 0$ corresponding to the eigenvector $\by_0 = \frac{1}{\sqrt{N(N+1)}}[1, \ 1,\ \hdots 1, \ -N] \in \mathbb{R}^{N+1}$.

We are interested in the change in the zero eigenvalue corresponding to the eigenvector $\by_0$.
It thus follows
\begin{align}
\begin{split} \label{eq:smallb2}
    f_b & = f_0 + c \by_0^\top \dfrac{P_i + P_i ^\top}{2} \by_0 = 0 + c (-0.1) = -0.1 c.
\end{split}
\end{align}

\textbf{The limit of large $c$ and unidirectional connection:}
Here, we set $\epsilon = 1$ and
\begin{equation}
    A_0 = c V^\top \frac{P_i+P_i^\top}{2}V = c V^\top P_i V, \quad A_1 = V^\top \dfrac{\tilde{L} + \tilde{L}^\top}{2} V
\end{equation}
where the matrices $P_i$ and $\tilde{L}$ are defined in Eq.\, \eqref{eq:lagematrix2}.
Now, the eigenvalue of interest is the most negative nonzero eigenvalue of $\frac{P_i+P_i^\top}{2}$ which is $f_0 =  -\frac{\sqrt{2}+1}{2}$.
The corresponding eigenvector is $\by_0^i = [0 \ \cdots \ 0 \ \ {\sqrt{2}+1} \ \ 0 \ \cdots \ 0 \ \ -1]^\top$ where the $i$th and the last entry of this eigenvectors have nonzero elements.
We set $\by_0^i = V\bv_0^i$ as the effective vector for the first-order approximation of matrix perturbation theory.
The first order change in the eigenvalue $f_0$ is ${\by_0^i}^\top A_1 \by_0^i = {\by_0^i}^\top \frac{\tilde{L} + \tilde{L}^\top}{2} \by_0^i = (3+2\sqrt{2})  L_{ii} > 0$.
Therefore, the perturbed eigenvalue after pinning node $i$ is
\begin{equation}
    f_b = f_0 + \epsilon {\by_0^i}^\top A_1 \by_0^i = -\frac{\sqrt{2}+1}{2}c + (3+2\sqrt{2})  L_{ii}.
\end{equation}

\subsection{Digraphs with a Bidirectional Connection} \label{sec:bi}

\textbf{The limit of small $c$ and bidirectional connection:} all the derivations are similar to the case of unidirectional connection except for the matrix
\begin{equation} \label{eq:pbi}
    P_i = 
    \begin{bmatrix}
    0 & \hdots & 0 & 0 & 0 & \hdots & 0 & 0 \\ 
    \vdots & \ddots & \vdots & \vdots & \vdots & \ddots & \vdots & \vdots \\ 
    0 & \hdots & 0 & 0 & 0 & \hdots & 0 & 0 \\ 
    0 & \hdots & 0 & -1 & 0 & \hdots & 0 & 1 \\ 
    0 & \hdots & 0 & 0 & 0 & \hdots & 0 & 0 \\ 
    \vdots & \ddots & \vdots & \vdots & \vdots & \ddots & \vdots & \vdots \\ 
    0 & \hdots & 0 & 0 & 0 & \hdots & 0 & 0 \\ 
    0 & \hdots & 0 & 1 & 0 & \hdots & 0 & -1
    \end{bmatrix}.
\end{equation}
This results in
\begin{align}
\begin{split} 
    f_b & = f_0 + c \by_0^\top \dfrac{P_i + P_i ^\top}{2} \by_0 = 0 + c (-1.1) = -1.1 c.
\end{split}
\end{align}

\textbf{The limit of large $c$ and bidirectional connection:}
Here, we set $\epsilon = 1$ and
\begin{equation}
    A_0 = c V^\top \frac{P_i+P_i^\top}{2}V = c V^\top P_i V, \quad A_1 = V^\top \dfrac{\tilde{L} + \tilde{L}^\top}{2} V
\end{equation}
where the matrices $P_i$ and $\tilde{L}$ are defined in Eqs.\,\eqref{eq:pbi} and \eqref{eq:lagematrix2}.
Now, the eigenvalue of interest is the one and only nonzero eigenvalue of $V^\top P_i V$ which is $f_0 = -2$.
Since $V$ is a similarity transformation for $N$ number of the eigenvalues, the matrix $P_i$ has the same eigenvalue $f_0 = -2$ and the corresponding eigenvector is $\by_0^i = [0 \ \cdots \ 0 \ \ {-1/\sqrt{2}} \ \ 0 \ \cdots \ 0 \ \ 1/\sqrt{2} ]^\top$ where the $i$th and the last entry of this eigenvectors have nonzero elements.
We define $\by_0^i := V\bv_0^i$ as the effective vector for the first-order approximation of matrix perturbation theory.
The first order change in the eigenvalue $f_0$ is ${\by_0^i}^\top A_1 \by_0^i = {\by_0^i}^\top \frac{\tilde{L} + \tilde{L}^\top}{2} \by_0^i = L_{ii} /2 > 0$.
Therefore, the perturbed eigenvalue after pinning node $i$ is
\begin{equation}
    f_b = f_0 + \epsilon {\by_0^i}^\top A_1 \by_0^i = -2c + \dfrac{L_{ii}}{2}.
\end{equation}

\subsection{Proof of Mean-Slope Relation} \label{sec:meanslope}

We show this relation $< df / d c > = -1/N$ in what follows.
{
We emphasize that we proceed without introducing the assumption that the network is  balanced.
}
Based on first order perturbation theory, $df_i / dc = \by_0^\top \frac{P_i + P_i^\top}{2} \by_0$ where $\by_0 = V \bv_0$ where $\bv_0$ is eigenvector corresponding to the smallest eigenvalue of the matrix $V^\top \frac{\tilde{L} + \tilde{L}^\top}{2}V$ and the matrices $P_i$ and $\tilde{L}$ are defined in Eq.\,\eqref{eq:lagematrix2}.
Note that $\by_0$ is in the range of the matrix $V$, resulting in the zero-sum of the entries of the vector $\by_0$.
It follows
\begin{align}
\begin{split}
    \left< \frac{d f}{d c} \right> = \dfrac{1}{N} \sum_{i=1}^N \frac{d f_i}{d c} & = \dfrac{1}{N} \by_0^\top \left( \sum_{i=1}^N  \dfrac{P_i + P_i^\top}{2}\right) \by_0 \\
    & =: \dfrac{1}{N} \by_0^\top \left( \bar{P}\right) \by_0,
\end{split}
\end{align}
where the matrix $\bar{P}$ has the following form
\begin{equation}
    \bar{P} = \sum_{i=1}^N  \dfrac{P_i + P_i^\top}{2} = \begin{bmatrix}
-1 &  &  &  & \frac{1}{2} \\ 
 & -1 & &  & \frac{1}{2} \\ 
 &  & \ddots &  & \vdots \\ 
 &  &  & -1 & \frac{1}{2} \\ 
\frac{1}{2} & \frac{1}{2} & \hdots & \frac{1}{2} & 0 
\end{bmatrix}.
\end{equation}
It follows that the matrix $\bar{P}$ has $N-1$ eigenvalues which are equal to $-1$ and the corresponding eigenvectors have their sum of the entries equal to zero. 
This suggests that these eigenvectors form a basis for the vector $\by_0$ and thus, the product $\by_0^\top \bar{P} \by_0 = -1$. 
Therefore, $< df / d c > = -1/N$. 
\hfill \qed

\subsection{Asymptotic Stability Analysis for the Case of Synchronization Dynamics \label{Sec:MAI}}

Next we present standard derivations from \cite{fujisaka1983stability,pecora1998master}.
A particular solution for the synchronization dynamics  \eqref{eq:synch} is the synchronous solution $\bx_1(t)=\bx_2(t)=...=\bx_N(t)=\bs(t)$, where $\bs(t)$, evolves in time based on the equation:
\begin{equation} \label{eq:s}
    \dot{\bs} (t) = \bF(\bs(t)).
\end{equation}
The synchronization manifold is the sets of phase space points that satisfy  $\bx_1 = \bx_2 = \hdots \bx_N$.

We study the time evolution of infinitesimal perturbations about the synchronization manifold, $\delta \bx_i (t) \coloneq \bx_i (t) - \bs(t)$. 
The synchronization is stable if $\delta \bx_i (t) \to 0$ as $t \to \infty$, $\forall i$, which implies $\bx_i (t) \to \bs(t)$ as $t \to \infty$, $\forall i$. 

To study the asymptotic stability of the perturbations $\delta \bx_i (t)$, we linearize Eq.\,\eqref{eq:synch} about $\bs(t)$ and derive the following linearized sets of dynamical equations:
\begin{equation} \label{eq:dx}
    \delta \dot{\bx}_i (t) = D\bF(\bs(t)) \delta \bx_i(t) - \sigma \sum_{j = 1}^{N} L_{ij} D\bH (\bs(t)) \delta \bx_j (t),
\end{equation}
$i = 1, \hdots, N$, where $D\bF (\bs(t))$ and $D \bH(\bs(t))$ are the Jacobians of $\bF$ and $\bH$ evaluated at $\bs(t)$, respectively.
A compact form of the above equations can be obtained by first concatenating the perturbation vectors as $\delta \bX(t) \coloneq \left[\delta \bx_1 (t)^\top, \hdots, \delta \bx_N (t)^\top \right]^\top \in \mathbb{R}^{mN}$, and then writing:
\begin{equation} \label{eq:compact}
    \delta \dot{\bX} (t) = \Bigg(I_N \otimes D \bF(s) - \sigma L \otimes D\bH(s)\Bigg) \delta \bX(t),  
\end{equation}
where $\otimes$ denotes the Kronecker product. We call $V$ the matrix having for columns the eigenvectors of $L$ and $\Lambda$ the associated diagonal matrix, having the eigenvalues of $L$ on the main diagonal.
Equation \eqref{eq:compact} can be decoupled through the eigen-decomposition of the Laplacian matrix $L$ by applying the similarity transformation $T = V \otimes I_m \in \mathbb{R}^{mN \times mN}$ as
\begin{equation} \label{eq:compact1}
    \delta \dot{\bZ} (t) = \Bigg(I_N \otimes D \bF(s) - \sigma \Lambda \otimes D\bH(s)\Bigg) \delta \bZ(t),  
\end{equation}
where $\delta \bZ (t) = T^{-1} \delta \bX (t) \in \mathbb{R}^{mN}$. 
By setting $[\delta \bz_1(t)^{\top}, \hdots , \delta \bz_N (t)^\top]^\top \coloneq \delta \bZ (t)$, with $\delta \bz_i (t) \in \mathbb{R}^m, \forall i$, we can write the decomposed equation of transformed perturbations as
\begin{equation} \label{dzi}
    \delta \dot{\bz}_i (t) = \Bigg(D\bF(\bs(t)) - \sigma \lambda_i D\bH(\bs (t)) \Bigg) \delta \bz_i (t),
\end{equation}
$i = 1, \hdots, N$. By construction, one eigenvalue $\lambda_1=0$, which is associated to stability of the dynamics on the synchronization manifold. As such it is excluded from the analysis that follows, that focus on the stability of the dynamics transverse to this manifold.
The synchronization is stable if $\delta \dot{\bz}_i (t) \to 0$ as $t \to 0$, $\forall i$.

From Eq.\ \eqref{dzi} we see that the stability of the synchronization is dependent on the Jacobians $D\bF(\bs(t))$ and $D\bH(\bs(t))$, and the scalar $\sigma \lambda_i$.
We introduce the parameter $K = \sigma \lambda_i$ and determine the stability of
\begin{equation} \label{eq:decomposed}
    \dot{\by} (t) = \Bigg(D\bF(\bs(t)) - K D\bH(\bs (t)) \Bigg) \by (t),
\end{equation}
where $\by (t) \in \mathbb{R}^m$ is a generic vector representing perturbations.
The stability is assessed through the evaluation of the Maximum Lyapunov Exponent (MLE),
\begin{equation}
    \text{MLE} = \lim_{t\to \infty} \dfrac{1}{t} \log \dfrac{\| \by(t) \|}{\| \by (0) \|}.
\end{equation}
If the MLE is negative (positive), the above system is stable (unstable).
Now, assume the set of parameters $\mathcal{K}$ for which the system \eqref{eq:decomposed} is stable (MLE $<0$).
If $\sigma \lambda_i \in \mathcal{K}$, $i=2,...,N$, then the synchronization dynamics \eqref{eq:synch} is stable and $\bx_i (t) \to \bs(t)$ when $t \to \infty$.
The master stability function (MSF) maps the parameter $K$ to the maximum Lyapunov exponent (MLE) of Eq.~\eqref{eq:decomposed}.

In the following, we present MSF plots for real values of $K$ for three well-known oscillators, the Lorenz, R\"ossler, and Chua systems, and for several choices of the output function $\mathbf{H}(\mathbf{x})$. For all combinations of $\mathbf{F}(\mathbf{x})$ and $\mathbf{H}(\mathbf{x})$ considered in this study, the MSF is positive for negative values of $K$. Under the assumption that this property holds in general, it follows that  a sufficient condition for the synchronous state to become unstable is that at least one of the eigenvalues of the Laplacian matrix $L$ has negative real part.

\textbf{Lorenz MSF plots}

The state vector, the local dynamics, and the coupling function for the Lorenz system are 
\begin{align}
    \bx = \begin{bmatrix}
        x \\ y \\ z
    \end{bmatrix}, \quad \bF (\bx) = \begin{bmatrix}
        10( y-x ) \\
        x( 28-z ) - y \\
        xy - \frac{8}{3} z
    \end{bmatrix}, \quad \bH(\bx) = H \bx,
\end{align}
where the matrix $H$ is the linear coupling matrix. 
We provide 9 MSF plots for the Lorenz system, where for each plot, a different entry of $H$ is equal to 1, and all others are zero.
It  follows that the Jacobian $D\bH(\bs) = H$.
The Jacobian of the local dynamics is
\begin{equation}
    \bs = \begin{bmatrix}
        x_s \\ y_s \\ z_s
    \end{bmatrix}, \quad D \bF (\bs) = \begin{bmatrix}
        -10 & 10 &  0 \\
        28-z_s& -1& -x_s \\
        y_s & x_s & -\frac{8}{3}
    \end{bmatrix}
\end{equation}

\begin{figure}
    \centering
    \includegraphics[width=0.99\linewidth]{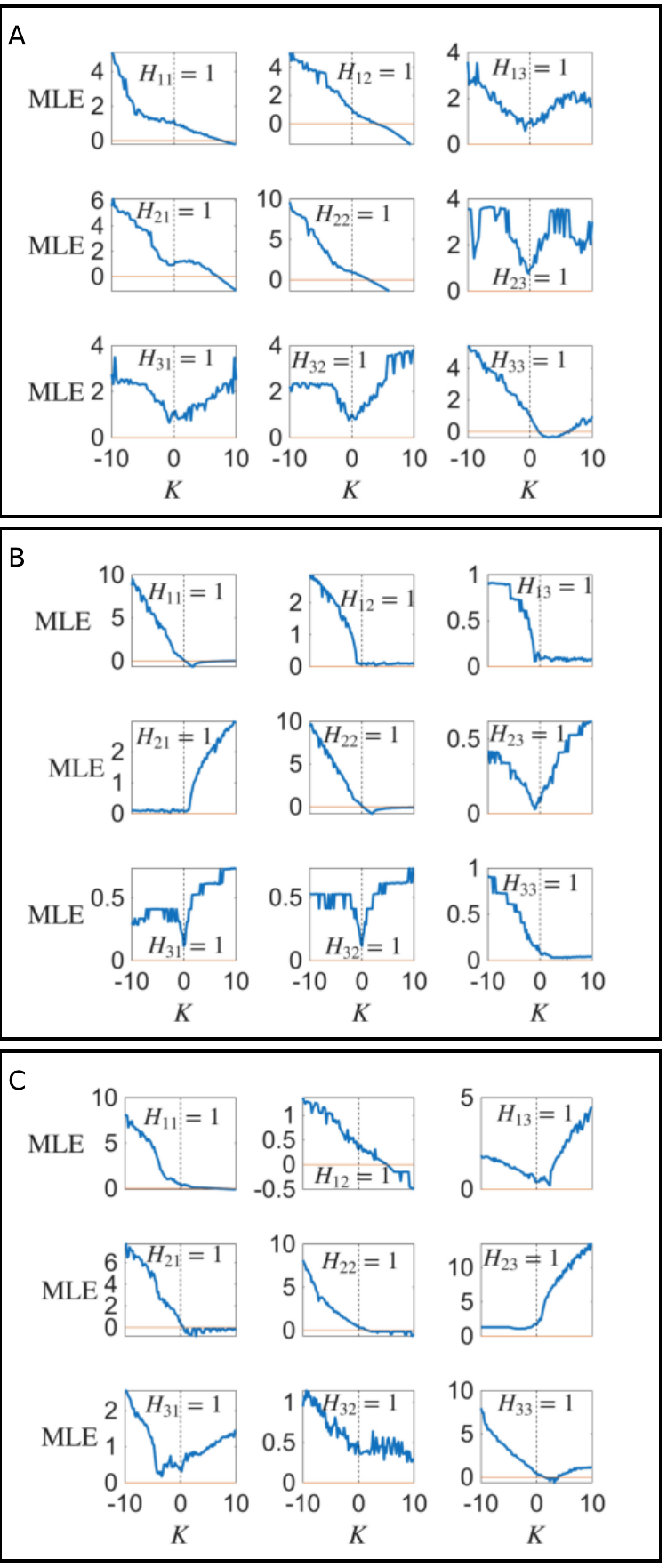}
    \caption{\textbf{Master stability functions for both positive and negative values of the function argument.}
    Master stability plots for the Lorenz system (A), Roessler system (B), and Chua System (C). Plots show the Maximum Lyapunov Exponent (MLE) as the argument $K$ of Eq.\,\eqref{eq:decomposed} is varied.
    In each subfigure, panel in row $i$ (from the top) and column $j$ (from the left) shows the case when $H_{ij} = 1$, and all other entries of the matrix $H$ are zero.
    }
    \label{fig:MSFLorenz}
\end{figure}

Figure \ref{fig:MSFLorenz} shows that the MSF is positive in all different coupling setups for the negative range of the MSF argument $K$. 
Thus, it shows that the synchronization is locally unstable for negative ranges of the parameter $K$.

\textbf{R\"ossler MSF plots}

The state vector, the local dynamics, and the coupling function for the Rossler system are 
\begin{align}
    \bx = \begin{bmatrix}
        x \\ y \\ z
    \end{bmatrix}, \quad \bF (\bx) = \begin{bmatrix}
        -y - z \\
        x + 0.1y \\
        0.1 + ( x-35 )z
    \end{bmatrix}, \quad \bH(\bx) = H \bx,
\end{align}
where the matrix $H$ is the linear coupling matrix. 
We provide 9 MSF plots for the Rossler system, where for each plot, a different entry of $H$ is equal to 1, and all others are zero.
It  follows that the Jacobian $D\bH(\bs) = H$.
The Jacobian of the local dynamics is
\begin{equation}
    \bs = \begin{bmatrix}
        x_s \\ y_s \\ z_s
    \end{bmatrix}, \quad D \bF (\bs) = \begin{bmatrix}
        0 & -1 & -1 \\
        1 & 0.1 & 0 \\
        z_s & 0 & x_s-35
    \end{bmatrix}
\end{equation}


Figure \ref{fig:MSFRossler} shows that the MSF is positive in all different coupling setups for the negative range of the MSF argument $K$. 
Thus, it shows that the synchronization is locally unstable for negative ranges of the parameter $K$.

\textbf{Chua's circuit MSF plots}

The state vector, the local dynamics, and the coupling function for the Chua's circuit system are 
\begin{align}
    \bx = \begin{bmatrix}
        x \\ y \\ z
    \end{bmatrix}, \quad \bF (\bx) = \begin{bmatrix}
        10( y - x + g(x) ) \\
        x - y + z \\
        -17.85y]
    \end{bmatrix}, \quad \bH(\bx) = H \bx,
\end{align}
where the matrix $H$ is the linear coupling matrix, and the function 
\begin{equation}
    g(x) = \begin{cases}
        - bx - a + b, & \quad x > 1, \\
        - ax, & \quad  |x| < 1, \\
        - bx + a - b, & \quad x < -1,
    \end{cases}, \quad \begin{cases}
        a = -1.2, \\
        b = -0.7.
    \end{cases}
\end{equation}
We provide 9 MSF plots for the Chua's circuit system, where for each plot, a different entry of $H$ is equal to 1, and all others are zero.
It  follows that the Jacobian $D\bH(\bs) = H$.
The Jacobian of the local dynamics is
\begin{equation}
    \bs = \begin{bmatrix}
        x_s \\ y_s \\ z_s
    \end{bmatrix}, \quad D \bF (\bs) = \begin{bmatrix}
        -10-10 d(x_s) & 10 & 0 \\
    1 & -1 & 1 \\
    0 & -17.85 & 0
    \end{bmatrix}
\end{equation}
where $d(x_s) = b$ when $|x_s| > 1$, and $d(x_s) = a$, otherwise.


Figure \ref{fig:MSFChua} shows that the MSF is positive in all different coupling setups for the negative range of the MSF argument $K$. 
Thus, it shows that the synchronization is locally unstable for negative ranges of the parameter $K$.

\subsection{Asymptotic Stability Analysis for the Case of Power Grid Dynamics \label{Sec:MAII}}

We linearize Eq.\  \eqref{eq:swingsimple} about the synchronous operating state $\theta_i^*= \omega_s t+ \theta_i^0$,
\begin{equation} \label{Eq:J}
\begin{aligned}
     \delta \ddot{\vartheta}_i (t) &= - {\gamma} {\delta \dot\vartheta}_i (t) - \sum_{j= 1, j \neq i}^N A_{ij} \cos(\theta_i^0 - \theta_j^0) (\delta \vartheta_i-\delta \vartheta_j)\\
     &= - {\gamma} {\delta \dot\vartheta}_i (t) - \sum_{j= 1, j \neq i}^N A'_{ij}  (\delta \vartheta_i-\delta \vartheta_j)\\
     &= - {\gamma} {\delta \dot\vartheta}_i (t) - \sum_{j= 1}^N L'_{ij}  \delta \vartheta_j
     \end{aligned}
     \end{equation}
$i=1,...,N$, where $\gamma>0$, $A=[A_{ij}]$ is the original adjacency matrix,  i.e.,  $A_{ij}=A_{ji}=1$ ($0$) if there is (is not) a transmission line between between nodes $i$ and $j$, $A'=[A'_{ij}]$ is the weighted adjacency matrix,
with entries $A'_{ij}= A_{ij} \cos(\theta_i^0 - \theta_j^0)$, and $L'=[L'_{ij}]$ is the associated Laplacian matrix, with entries $L'_{ij}= (\delta_{ij} \sum_j A'_{ij}  - A'_{ij})$. We proceed under the assumption that, prior to the intervention of the intruder, the differences $\theta_i^0 - \theta_j^0$ are small, so that  $A'_{ij}>0$ iff $A_{ij}=1$, for which the Laplacian matrix $L'$ is proper.

By diagonalizing the Laplacian matrix $L'$, the last line of Eq.\ \eqref{Eq:J} can be transformed into the following independent equations,
\begin{equation} \label{Eq:JJ}
    \ddot{y}_k(t)= - {\gamma} {\dot y}_k (t) - \lambda'_{k}  y_k(t),
\end{equation}
$k=1,..,N$, where $\lambda'_{k}$, $k=1,..,N$,  are the eigenvalues of the matrix $L'$. By construction there is one eigenvalue $\lambda'_{k}=0$, which corresponds to marginal stability of the synchronous operating state $\theta_i^*= \omega_s t+ \theta_i^0$.
Transverse stability is determined by the remaining eigenvalues of $L'$, as we explain below. 

In the case that $\lambda'_k$ is real, Eq.\ \eqref{Eq:JJ} is: (i) stable for $\lambda'_{k}>0$, (ii) unstable for $\lambda'_{k}<0$, and (iii) marginally stable for $\lambda'_{k}=0$. 

A more interesting case is for a complex eigenvalue of the Laplacian $\lambda'_k=\alpha+ j \beta$, for which the characteristic equation is,
\begin{equation} \label{pol}
    r^2+\gamma r +(\alpha+j \beta)=0,
\end{equation}
where here $r$ is a complex number and we recall that $\gamma$ is real and positive. The largest real part root of Eq.\ \eqref{pol} is
\begin{equation} \label{largest_real_root}
    -\frac{\gamma}{2} + \frac{1}{2} \sqrt{\frac{|D|+\gamma^2 - 4 \alpha}{2}},
\end{equation}
where the modulus of the complex discriminant $|D|=\sqrt{(\gamma^2-4 \alpha)^2 + 16 \beta^2}$. Eq.\ \eqref{largest_real_root} is negative if 
\begin{equation} \label{crilu}
\beta^2< \gamma^2 \alpha.
\end{equation}
For stability of the synchronous operating state, criterion \eqref{crilu} needs to be satisfied for all the eigenvalues $\lambda'_k=\alpha+ j \beta$ of $L'$ (with the exception of the one zero eigenvalue.) It is easy to see that \eqref{crilu} is never satisfied for $\alpha<0$. 

We thus conclude that a sufficient condition for the synchronous operating state to be unstable is that at least one of the eigenvalues of the  the matrix $L'$ has negative real part (however, instability is also possible when all the eigenvalues of $L'$ have negative or zero real part.)

\subsection{Asymptotic Stability Analysis for the Case of Formation Control Dynamics \label{Sec:MAIII}}

In the case of Eq.\ \eqref{Eq_formation_control}, the dynamics is linear and stability depends uniquely on the spectrum of the matrix $F(L)$. In the case considered in this paper that $F(L) = I_N - \exp (-3L)$, each eigenvalues of the matrix $F(L)$ is equal to $\mu_k=1- \exp (-3 \lambda_k)$, where $\lambda_k$ is an eigenvalue of $L$, $k=1,..,N$. 

It follows that the transformation that maps $\lambda_k$ into  $\mu_k$ preserves the origin and the sign of the real part of the eigenvalues. Specifically,
\begin{equation}
\lambda_k = 0 \iff \mu_k = 0,
\end{equation}
and
\begin{equation}
\operatorname{sgn}\!\left(\Re(\mu_k)\right)
=
\operatorname{sgn}\!\left(\Re(\lambda_k)\right),
\qquad k = 1,\ldots,N,
\end{equation}
where with $\operatorname{sgn}\!(x)$ we indicate the sign of $x$.

We conclude that if the Laplacian matrix $L$ is (is not) proper, the formation control state is stable (unstable.)

\subsection{Asymptotic Stability Analysis for the Case of the Kuramoto Dynamics \label{Sec:MAIIII}}

We fix the adjacency matrix $A$ and linearize Eq.\  \eqref{eq:swingsimple} about the phase locked state $\theta_i^*= \omega_s t+ \theta_i^0$,
\begin{equation} \label{Eq:K}
\begin{aligned}
     \delta \dot{\vartheta}_i (t) &= -\dfrac{1}{N+1} \sum_{j=1, j \neq i}^{N+1}  A_{ij} \cos(\theta_i^0 - \theta_j^0) (\delta \vartheta_i-\delta \vartheta_j)\\
     &=  -\dfrac{1}{N+1} \sum_{j=1, j \neq i}^{N+1}  A'_{ij}  (\delta \vartheta_i-\delta \vartheta_j)\\
     &= -\dfrac{1}{N+1}  \sum_{j= 1, j \neq i}^{N+1} L'_{ij}  \delta \vartheta_j
     \end{aligned}
     \end{equation}
$i=1,...,N$, where  $A=[A_{ij}]$ is the fixed adjacency matrix defined in Eq. \eqref{eq:Akuramoto} either before or after the attack, $A'=[A'_{ij}]$ is the weighted adjacency matrix,
with entries $A'_{ij}= A_{ij} \cos(\theta_i^0 - \theta_j^0)$, and $L'=[L'_{ij}]$ is the associated Laplacian matrix, with entries $L'_{ij}= (\delta_{ij} \sum_j A'_{ij}  - A'_{ij})$. We proceed under the assumption that, prior to the intervention of the intruder, the differences $\theta_i^0 - \theta_j^0$ are small, so that  $A'_{ij}>0$ iff $A_{ij}=1$, for which the Laplacian matrix $L'$ is proper.

By diagonalizing the Laplacian matrix $L'$, the last line of Eq.\ \eqref{Eq:K} can be transformed as follows,
\begin{equation} \label{Eq:KK}
    \dot{y}_k= - \lambda'_{k}  y_k,
\end{equation}
$k=1,..,N$, where $\lambda'_{k}$, $k=1,..,N$,  are the eigenvalues of the matrix $L'$. Equation \eqref{Eq:KK} is stable if the real part of $\lambda_k'$ is positive, unstable if it is negative, and marginally stable if $\lambda_k'=0$. By construction there is one eigenvalue $\lambda'_{k}=0$, which corresponds to marginal stability of the phase locked state $\theta_i^*= \omega_s t+ \theta_i^0$. Transverse stability is achieved if all the other eigenvalues of $L'$ have positive real part. 

We conclude that a sufficient condition for the phase-locked state to become unstable is that at least one of the eigenvalues of the  the matrix $L'$ has negative real part.

\section{Reproducibility Code Description}

The MATLAB code used to generate Fig.~\ref{fig:synch_node}B performs the following steps. 
First, the adjacency matrix $A$ is loaded, and the graph Laplacian matrix $L$ is constructed.
For each selected target node $i$, an attack vector $\bb$ is defined with $b_i=-c$ and all other entries equal to zero, with $c=1$ as the budget for the attack. 
The unidirectional augmented Laplacian after the attack is then constructed as in Eq.~\eqref{eq:Lnewuni}.
The simulation uses the original Laplacian before the attack time $t_s=4 s$ and switches to the augmented Laplacian for $t \geq t_s$. 
A network of coupled Lorenz oscillators is then simulated using MATLAB’s ODE solver. 
The norm of the transverse perturbations $\| \delta \bX (t) \|$ is evaluated as in Eq.~\eqref{eq:trans_synch}. 
Finally, the code plots $\| \delta \bX (t) \|$ in time for different attacked nodes, and saves a figure similar to Fig.~\ref{fig:synch_node}B.

\color{black}

\section*{Data availability}
All data generated or analyzed during this study are included in this published article (and its supplementary information files).

\section*{Code availability}

The MATLAB code used in this study is openly available as a reproducible compute capsule on Code Ocean at \url{https://codeocean.com/capsule/0183273/tree/v1} under an MIT license.

\color{black}

\section*{Funding Statement}
FS acknowledges support from grants AFOSR FA9550-24-1-0214 and Oak Ridge National Laboratory 006321-00001A. MA acknowledges funding support from FSU CRS-SEED program, `Structure and Dynamics of Non-Normal Networks.' HAM acknowledges funding  from NSF-HNDS Grant 2214217.

\section*{Author Contributions}
AN worked on the theory and numerical simulations. ST worked on the numerical simulations and the figures. MA worked on the theory and contributed to writing the paper. DP and HAM provided essential input and guidance.    FS worked on the theory and supervised the research.

\section*{Competing Interests Statement}
The authors declare no competing interests

\end{document}